\documentclass[bj,preprint]{imsart}

\RequirePackage[colorlinks,citecolor=blue,urlcolor=blue]{hyperref}

\usepackage{bm}
\hypersetup{pdfauthor={Jochmans, Henry, and Salani\'e}}
\hypersetup{pdftitle={Inference on mixtures}}

\startlocaldefs
\endlocaldefs
\usepackage{endnotes} \let\footnote=\endnote

\usepackage{amsmath}
\usepackage{amsthm}
\usepackage{amsfonts}
\usepackage{amssymb}
\usepackage{appendix}
\usepackage{graphicx} 
\usepackage{mathrsfs}
\usepackage{lscape}
\usepackage[multiple]{footmisc}

\theoremstyle{plain} \newtheorem{assumption}{\mdseries\scshape Assumption}
\theoremstyle{plain} 
\theoremstyle{plain} \newtheorem{theorem}{\mdseries\scshape Theorem}
\theoremstyle{plain} \newtheorem{corollary}{\mdseries\scshape Corollary}
\theoremstyle{plain} 
\theoremstyle{plain} \newtheorem{example}{\mdseries\scshape Example}

\usepackage{sectsty}
\sectionfont{\footnotesize\uppercase} 
\subsectionfont{\normalsize} 

\usepackage{natbib} \RequirePackage[colorlinks,citecolor=blue,urlcolor=blue]{hyperref}
\numberwithin{equation}{section}

\begin{document}

\begin{frontmatter}

\title{\textbf{\normalsize INFERENCE ON TWO-COMPONENT MIXTURES UNDER TAIL RESTRICTIONS\thanksref{H}}}
\runtitle{Inference on mixtures}

\begin{aug}
\author{\fnms{\small Koen} \snm{\small Jochmans}\thanksref{a}\ead[label=e1]{koen.jochmans@sciencespo.fr}}
\and
\author{\fnms{\small Marc} \snm{\small Henry}\thanksref{b}\ead[label=e2]{marc.henry@psu.edu}}
\and
\author{\fnms{\small Bernard} \snm{\small Salani\'e}\thanksref{c}\ead[label=e3]{bsalanie@columbia.edu}}
\address[a]{Sciences Po, 28 rue des Saints P\`eres, 75007 Paris, France. \\ \printead{e1}} 
\address[b]{The Pennsylvania State University, University Park, PA 16801, U.S.A. \\ \printead{e2}} 
\address[c]{Columbia University, 420 West 118th Street, New York, NY 10027, U.S.A. \\ \printead{e3}} 

\thankstext{H}{We are grateful to Peter Phillips, Arthur Lewbel, and three referees for comments and suggestions, and to Victor Chernozhukov and Yuichi Kitamura for fruitful discussions. Parts of this paper were written while Henry was visiting the University of Tokyo Graduate School of Economics and while Salani\'e was visiting the Toulouse School of Economics. The hospitality of both institutions is gratefully acknowledged. Jochmans' research has received funding from the SAB grant ``Nonparametric estimation of finite mixtures''. Henry's research has received funding from the SSHRC Grants 410-2010-242 and 435-2013-0292, and NSERC Grant 356491-2013. Salani\'e thanks the Georges Meyer endowment. Some of the results presented here previously circulated as part of  \cite{HenryKitamuraSalanie2010}, whose published version (\citealt{HenryKitamuraSalanie2014}) only contains results on partial identification.}  

\runauthor{K.~Jochmans, M.~Henry, and B.~Salani\'e} 
\affiliation{Sciences Po \and Penn State \and Columbia}
\end{aug}

\begin{abstract}
Final version: \today 

\medskip\noindent Many econometric models can be analyzed as f\mbox{}inite mixtures. We focus on two-component mixtures and we show that they are nonparametrically point identif\mbox{}ied by a combination of an exclusion restriction and tail restrictions. Our identif\mbox{}ication analysis suggests simple closed-form estimators of the component distributions and mixing proportions, as well as a specification test. We derive their asymptotic properties using results on tail empirical processes and we present a simulation study that documents their f\mbox{}inite-sample performance.
\begin{keyword} 
\kwd{mixture model} \kwd{nonparametric identification and estimation} 
\kwd{tail empirical process} 
\end{keyword}

\end{abstract}

\end{frontmatter}

\baselineskip=0.6cm
\thispagestyle{empty}

\section*{Introduction} 

The use of f\mbox{}inite mixtures has a long history in applied econometrics. A non--exhaustive list of applications includes models with discrete unobserved heterogeneity, hidden Markov chains, and
models with mismeasured discrete variables; see \cite{HenryKitamuraSalanie2014} for a more extensive discussion of applications. Until recently, the literature on nonparametric identif\mbox{}ication of mixture models was sparse. Following the lead of \cite{HallZhou2003}, several authors have analyzed multivariate mixtures; recent contributions are \cite{KasaharaShimotsu2009}, \cite{AllmanMatiasRhodes2009}, and \cite{BonhommeJochmansRobin2014b,BonhommeJochmansRobin2014a}. There are fewer identifying restrictions available when the model of interest is univariate. \cite{BordesMotteletVandekerkhove2006}, for instance, provide such restrictions for location models with symmetric error distributions.

In this paper we give suf\mbox{}f\mbox{}icient conditions that point-identify univariate component distributions and associated mixing proportions. The restrictions we rely on are most ef\mbox{}fective in two-component models; and to simplify the analysis, we focus on this case, like \cite{HallZhou2003} and
\cite{BordesMotteletVandekerkhove2006}. We comment briefly on mixtures with more components at the end of the paper. Our arguments are constructive, and we propose closed-form estimators for both the component distributions and the mixing proportions. We derive their large-sample
properties and we propose a specification test. Finally, we investigate the behavior of our inference tools  in a simulation experiment.

The model we consider in this paper is characterized by an exclusion restriction and a tail-dominance assumption. Like  \cite{HenryKitamuraSalanie2014}, we assume the existence of a source of variation
that shifts the mixing proportions but leaves the component distributions unchanged. Such an assumption is natural in several important applications, such as  measurement-error models (\citealt{Mahajan2006}), for example. In hidden Markov models, it follows directly from the model specif\mbox{}ication. The exclusion restriction is also implied by the conditional-independence restriction that underlies the results of \cite{HallZhou2003} and others on multivariate mixtures.

\cite{HenryKitamuraSalanie2014} have shown that our exclusion restriction implies that both the mixing proportions and the component distributions lie in a non-trivial set. However, they only proved partial identification, and they did not discuss inference.  Here we achieve point-identif\mbox{}ication by complementing the exclusion restriction with a restriction on the relative tail behavior of the component distributions. This restriction is quite natural in location models, for instance, but it can be motivated more generally. Regime-switching models typically feature regimes with dif\mbox{}ferent tail behavior, for example. Alternatively, theoretical models can imply the required tail behavior; an example is the search and matching model of \cite{ShimerSmith2000}, as explained in \cite{DhaultfoeuilleFevrier2014}.

Our identif\mbox{}ication argument suggests plug-in estimators of the mixing proportions and the component distributions that are available in closed form. The estimators are based on ratios of intermediate
quantiles, and their convergence rate is determined by the theory of tail empirical processes. As we rely on the tail behavior of the component distributions to infer the mixing proportions, our estimators converge more slowly than the
parametric rate. If the mixing proportions were known---or could be estimated at the parametric rate---the tail restrictions could be dispensed with and the implied estimator of the component distributions would also converge at the parametric rate.

Our estimators are consistent under very weak tail dominance assumptions. To control for asymptotic bias in their limit distribution we need to impose stronger requirements that prevent the tails of the components from vanishing too quickly. These assumptions rule out the Gaussian location model. Such thin-tailed distributions are known to be problematic for inference techniques that rely on tail behavior (\citealt{KhanTamer2010}). However, we show that our assumptions apply to distributions with fatter tails, such as Pareto distributions.

Identif\mbox{}ication only requires that the variable subject to the exclusion restriction can take on two values. If it can take on more values, the model is overidentif\mbox{}ied and the specif\mbox{}ication can be tested.

The tail conditions we use to obtain nonparametric identif\mbox{}ication are related to the well-known identif\mbox{}ication-at-inf\mbox{}inity argument of \cite{Heckman1990}; see also \cite{DhaultfoeuilleMaurel2013} for another approach. Other types of support restrictions have been used in related problems to establish identif\mbox{}ication. \cite{SchwarzVanBellegem2010} imposed support restrictions in a semiparametric deconvolution problem to deal with measurement error in location models. \cite{DhaultfoeuilleFevrier2014} relied on a support condition as an alternative to completeness conditions (\citealt{HuSchennach2008}) in multivariate mixture models.

The remainder of the paper is organized as follows. Section~\ref{sec:modelident} describes the mixture model and proves identif\mbox{}ication. We rely on these results to construct estimators and derive their asymptotic properties in Section~\ref{sec:estimators}. We also discuss specif\mbox{}ication testing at this point. In Section~\ref{sec:MC} we conduct a Monte Carlo experiment that gives evidence on the small-sample performance of our methods. Finally, we conclude with some remarks on mixtures with more than two components.

\section{Mixtures with exclusion and tail restrictions}
\label{sec:modelident}

Let $(Y,X)\in \mathbb{R}\times \mathcal{X}$ be random variables. We assume throughout that our mixtures satisfy the following simple exclusion restriction.\footnote{We omit conditioning variables throughout. The identif\mbox{}ication analysis extends straightforwardly. In principle, the distribution theory could be extended by using local empirical process results along the lines of \cite{EinmahlMason1997}. We postpone a detailed investigation into such an extension to future work.}
  
\begin{assumption}[Mixture with exclusion] \label{ass:model}
	$F(y\vert x)\equiv \mathbb{P}(Y\leq y\vert X=x)$ decomposes as
	the two-component mixture
	\begin{equation} \label{eq:model}
	F(y\vert x) = G(y) \, \lambda(x) + H(y) \, (1-\lambda(x))
	\end{equation}
	for distribution functions $G:\mathbb{R}\mapsto[0,1]$ and
	$H:\mathbb{R}\mapsto[0,1]$ and a function
	$\lambda:  \mathcal{X}\mapsto[0,1]$ that maps values $x$
	into mixing proportions.
\end{assumption}

\noindent
The assumption that the component distributions do not depend on $X$ embodies our exclusion restriction; see also \cite{HenryKitamuraSalanie2014}.

We complete the mixture model with the following assumption.
\begin{assumption} \label{ass:relevance}
The mixing proportion $\lambda$ is non-constant on $\mathcal{X}$ and is bounded away from zero and one on $\mathcal{X}$.
\end{assumption}

\noindent
Non-constancy of $\lambda$ gives the variable $X$ relevance. Bounding $\lambda$ away from zero and one implies that the mixture is irreducible.\footnote{Note that irreducibility rules out the possibility of achieving identif\mbox{}ication of $G$ and $H$ via an identif\mbox{}ication-at-inf\mbox{}inity argument, as in \cite{Heckman1990} and \cite{AndrewsSchafgans1998} for instance.}


\subsection{Motivating examples}

Our f\mbox{}irst example has a long history in empirical work (\citealt{Frisch1934}).

\begin{example}[Mismeasured treatments] \label{example:treatment}
Let $T$ denote a binary treatment indicator. Suppose that $T$ is
subject to classif\mbox{}ication error: rather than observing $T$,
we observe misclassif\mbox{}ied treatment $X$. The distribution of the
outcome variable $Y$ given $X=x$ is
\begin{equation*}
\begin{split}
F(y\vert x) =
\mathbb{P}(Y\leq y\vert T=1,X=x) \, \lambda(x)
+
\mathbb{P}(Y\leq y\vert T=0,X=x) \, (1-\lambda(x)),
\end{split}
\end{equation*}
with $\lambda(x)=\mathbb{P}(T=1\vert X=x)$. The usual ignorability assumption states that $X$ and $Y$ are independent given $T$. That is,
$$
\mathbb{P}(Y\leq y\vert T=t,X=x) = 
\mathbb{P}(Y\leq y\vert T=t),
$$ for $t\in\lbrace 0,1\rbrace$, in which case the decomposition of
$F(y\vert x)$ reduces to the model in \eqref{eq:model} with
$G(y)=\mathbb{P}(Y\leq y\vert T=1)$ and $H(y)= \mathbb{P}(Y\leq y\vert
T=0)$. Also note that $\lambda$ is non-constant unless
misclassif\mbox{}ication in $T$ is completely random.
\end{example}

\noindent
The identif\mbox{}ication of treatment ef\mbox{}fects when the treatment indicator is mismeasured has received considerable attention, especially in the context of regression models (\citealt{Bollinger1996}; \citealt{Mahajan2006}; \citealt{Lewbel2006}). Here, the conditional ignorability assumption that validates our exclusion restriction relies on non-dif\mbox{}ferential misclassif\mbox{}ication error. It has been routinely used elsewhere (\citealt{CarrollRuppertStefanskiCrainiceanu2006}).

Our second example deals with regime-switching models, also referred to as hidden Markov models. These models cover switching regressions, which have been used in a variety of settings (see, e.g., \citealt{Heckman1974}, \citealt{Hamilton1989}), as well as several versions of stochastic-volatility models
(\citealt{GhyselsHarveyRenault1996}).

\begin{example}[Hidden Markov model] \label{example:HMM}
Let $\bm{Y}=(Y_1,\ldots,Y_T)^\prime$
 be a time series of outcome variables. A
hidden Markov model for the dependency structure in these data assumes
that there is a discrete latent series of state variables
$\bm{S}=(S_1,\ldots,S_T)^\prime$ having Markovian dependence, that the
variables in $\bm{Y}$ are jointly independent given
$\bm{S}$, and that
$$
\mathbb{P}(Y_t\leq y_t \vert \bm{S} = \bm{s}) = 
\mathbb{P}(Y_t\leq y_t \vert S_t = s_t).
$$ To see that such a model f\mbox{}its \eqref{eq:model}, assume that there
are two latent states 0 and 1 and (for notational simplicity) that $\bm{S}$
has f\mbox{}irst-order Markov dependence. Denote 
 $\bm{X}=(Y_1,\ldots,Y_{t-1})^\prime$.  Then
$$
F(y_t\vert \bm{x})
=
\mathbb{P}(Y_t\leq y_t \vert S_t=1) \, \mathbb{P}(S_t=1\vert \bm{X}=\bm{x} )
+
\mathbb{P}(Y_t\leq y_t \vert S_t=0) \, \mathbb{P}(S_t=0\vert \bm{X}=\bm{x} ),
$$ which f\mbox{}its our setup. Moreover,
$\lambda(\bm{x})=\mathbb{P}(S_t=1\vert \bm{X}=\bm{x})$ does vary with
$\bm{x}$, unless the outcomes are independent of the latent states.

\end{example}

\noindent
In this example, the exclusion restriction follows directly from the Markovian structure of the regime-switching model. \cite{GassiatRousseau2014} obtained nonparametric identif\mbox{}ication in location models when the matrix of transition probabilities of the Markov chain has full rank. The approach presented here delivers nonparametric identif\mbox{}ication in a much broader range of models.

Our third example links \eqref{eq:model} to the recent literature on multivariate mixtures that builds on \cite{HallZhou2003}.

\begin{example}[Multivariate mixture] \label{example:HallZhou}
Suppose $Y$ and $X$ are two measurements that are independent conditional on a latent binary factor $T$:
\begin{equation*}
\begin{split}
\mathbb{P}(Y\leq y, X\leq x)
& =
\mathbb{P}(Y\leq y\vert T=1) 
\, 
\mathbb{P}(X\leq x \vert T=1) 
\, \mathbb{P}(T=1)
\\
& + \, 
\mathbb{P}(Y\leq y\vert T=0) 
\, 
\mathbb{P}(X\leq x \vert T=0) 
\, \mathbb{P}(T=0).
\end{split}
\end{equation*}
Then the conditional distribution of the $Y$ given $X$ is 
$$
F(y\vert x) = 
\mathbb{P}(Y\leq y\vert T=1) \, \mathbb{P}(T=1\vert X=x) +
\mathbb{P}(Y\leq y\vert T=0) \, \mathbb{P}(T=0\vert X=x) .
$$
This is of the form in \eqref{eq:model} with $G(y)=\mathbb{P}(Y\leq y\vert T=1)$, $H(y)=\mathbb{P}(Y\leq y\vert T=0)$, and $\lambda(x) = \mathbb{P}(T=1\vert X=x)$. Note that the bivariate mixture model implies that the distribution of $X$ given $Y$ decomposes in the same way.
\end{example}

\noindent
\cite{HallZhou2003} showed that multivariate two-component mixtures with conditional independence restrictions are nonparametrically identif\mbox{}ied from data on three or more measurements
and are set identif\mbox{}ied from data on only two measurements. The results we derive below imply that two measurements can also yield point identif\mbox{}ication under tail restrictions.

\subsection{Identification}

We show below that both the mixture components $G, H$ and the mixing proportions $\lambda$ are identif\mbox{}ied under the following dominance condition on the tails of the component distributions.

\begin{assumption}[Tail dominance] \label{ass:dominance} \mbox{} \\
(i) The left tail of $G$ is thinner than the left tail of $H$, i.e.,  
$$
\lim_{y\downarrow-\infty}\frac{G(y)}{H(y)} = 0.
$$  
(ii) The right tail of $G$ is thicker than the right tail of $H$, i.e., 
$$
\lim_{y\uparrow  +\infty}\frac{1-H(y)}{1-G(y)} = 0.
$$
\end{assumption}

\bigskip


Tail dominance is natural in location models.

\begin{example}[Location models] \label{lemma:locationmodel}
Suppose that $Y = \mu(T)+U$, where $T$ is a binary indicator and
 $U\sim F$, independent of $T$. Then \eqref{eq:model} yields
$$
F(y\vert x) = F(y-\mu(1)) \, \mathbb{P}(T=1\vert X=x) + F(y-\mu(0)) \, \mathbb{P}(T=0\vert X=x).
$$
Suppose that $\mu(0)<\mu(1)$, that $F$ is  absolutely continuous, and that its hazard rate  $f(u)/
\big(1-F(u)\big)$  (resp.\ $f(u)/ F(u)$) goes to $+\infty$ as $u\uparrow
+\infty$  (resp.\ $u\downarrow -\infty$). Then Assumption
\ref{ass:dominance} holds with $G(y)= F(y-\mu(1))$ and $H(y)=F(y-\mu(0))$.
\end{example}

\begin{proof}
Let us show that Assumption \ref{ass:dominance}(ii) holds.  Let $\varphi(u)\equiv -\ln (1-F(u))$ and note that
$\varphi^\prime(u)=f(u)/(1-F(u))$. Then
$$
\frac{1-F(y-\mu(0))}{1-F(y-\mu(1))} =  \exp \left( \varphi(y-\mu(1))-\varphi(y-\mu(0))\right)
=
\exp \left( - \varphi^\prime(y^\ast) \, (\mu(1)-\mu(0)) \right)
$$ for some $y^\ast$  between $y-\mu(1)$ and $y-\mu(0)$. Since $\mu(1)>\mu(0)$ and
the hazard rate increases without bound as $y\uparrow +\infty$, the expression on the right-hand side tends to zero
as $y$ increases. Assumption \ref{ass:dominance}(i) can be verif\mbox{}ied in
the same way.
\end{proof}

\noindent
It is important to note that, aside from regularity conditions, we do not impose any shape restrictions on the mixture components outside of the tails.

We now show that, combined, our exclusion restriction and tail-dominance assumption identif\mbox{}y all elements of the mixture model.

\begin{theorem}[Identification] \label{thm:identification}
Under Assumptions \ref{ass:model}---\ref{ass:dominance},  $G$, $H$, and $\lambda$ are identif\mbox{}ied.
\end{theorem}

\begin{proof}
The proof is constructive. Fix $x^\prime\in\mathcal{X}$ and choose
$x^{\prime\prime}\in\mathcal{X}$ so that $\lambda(x^\prime) \neq
\lambda(x^{\prime\prime})$. Then re-arranging \eqref{eq:model} gives 
\begin{equation*}
\begin{split}
\frac{F(y\vert x^\prime)}{F(y\vert x^{\prime\prime})}
& =
\frac{1+\lambda(x^\prime)  \left(G(y)/H(y)-1\right)}{1+\lambda(x^{\prime\prime})  \left(G(y)/H(y)-1\right)},
\\
\frac{1-F(y\vert x^\prime)}{1-F(y\vert x^{\prime\prime})}
& =
\frac{\lambda(x^{\prime}) +  \left((1-H(y))/(1-G(y))\right) \, \left(1-\lambda(x^{\prime})\right)}{\lambda(x^{\prime\prime}) +  \left((1-H(y))/(1-G(y))\right) \, \left(1-\lambda(x^{\prime\prime})\right)}.
\end{split}
\end{equation*}
Taking limits, Assumption \ref{ass:dominance} further implies that
\begin{equation} \label{eq:zetas}
\begin{split}
\zeta^{-}(x^\prime,x^{\prime\prime})  \equiv \lim_{y\downarrow -\infty} \frac{F(y\vert x^\prime)}{F(y\vert x^{\prime\prime})} = \frac{1-\lambda(x^\prime)}{1-\lambda(x^{\prime\prime})},
\\
\qquad
\zeta^{+}(x^\prime,x^{\prime\prime})  \equiv \lim_{y\uparrow +\infty} \frac{1-F(y\vert x^\prime)}{1-F(y\vert x^{\prime\prime})} = \frac{\lambda(x^\prime)}{\lambda(x^{\prime\prime})} .
\end{split}
\end{equation}
These two equations can be solved for the mixing proportion at $x^\prime$, yielding
\begin{equation} \label{eq:lambda}
\lambda(x^{\prime}) = \frac{1-\zeta^{-}(x^{\prime\prime},x^{\prime}) }{\zeta^{+}(x^{\prime\prime},x^{\prime}) -\zeta^{-}(x^{\prime\prime},x^{\prime}) }.
\end{equation}
Since $\lambda$ is non-constant, for any $x^\prime\in\mathcal{X}$ there exists a $x^{\prime\prime}\in\mathcal{X}$ for which such a system of equations can be constructed. The function $\lambda$ is therefore identif\mbox{}ied on its entire support.  To establish identif\mbox{}ication of $G$ and $H$, f\mbox{}irst note that
\begin{equation} \label{eq:diff}
G(y)-H(y) = \frac{F(y\vert x^{\prime\prime}) - F(y\vert x^{\prime})}{\lambda(x^{\prime\prime})-\lambda(x^\prime)}
\end{equation}
follows from \eqref{eq:model}. Then, evaluating \eqref{eq:model} in $x^{\prime\prime}$ and re-arranging the resulting expression for $F(y\vert x^{\prime\prime})$ gives
$$
H(y) 
= F(y\vert x^{\prime\prime}) -  \big(G(y)-H(y)\big) \, \lambda(x^{\prime\prime})
= F(y\vert x^{\prime\prime}) - \frac{\lambda(x^{\prime\prime})}{\lambda(x^{\prime\prime})-\lambda(x^{\prime})} \big(F(y\vert x^{\prime\prime})-F(y\vert x^{\prime})\big),
$$
which is identif\mbox{}ied. Furthermore, using \eqref{eq:zetas} we can write
\begin{equation} \label{eq:H}
H(y) = F(y\vert x^{\prime\prime}) - \frac{1}{1-\zeta^{+}(x^{\prime},x^{\prime\prime})} \big(F(y\vert x^{\prime\prime})-F(y\vert x^{\prime})\big).
\end{equation}
Plugging this expression for $H(y)$ back into the mixture representation of $F(y\vert x^{\prime\prime})$ as in \eqref{eq:model} further yields
\begin{equation} \label{eq:G}
G(y) = F(y\vert x^{\prime\prime}) - \frac{1}{1-\zeta^{-}(x^{\prime},x^{\prime\prime})} \big(F(y\vert x^{\prime\prime})-F(y\vert x^{\prime})\big),
\end{equation}
again using \eqref{eq:zetas}. This shows that  both component distributions are identif\mbox{}ied,  concluding the proof.
\end{proof}

If we only assume one-sided tail dominance, then either $G$ or $H$ remains identif\mbox{}ied.

\begin{corollary}[One-sided tail dominance] \label{cor:onesided}
Under Assumptions \ref{ass:model} and  \ref{ass:relevance}, $G$ is identif\mbox{}ied if Assumption \ref{ass:dominance}(i) holds and $H$ is identif\mbox{}ied if Assumption \ref{ass:dominance}(ii) holds.
\end{corollary}

\begin{proof}
We consider identif\mbox{}ication of $H$. Let $x^\prime,x^{\prime\prime}$ be as in the proof of Theorem \ref{thm:identification}. Under Assumption \ref{ass:dominance}(ii) we can still determine
$\zeta^{+}(x^\prime,x^{\prime\prime})=\lambda(x^{\prime})/\lambda(x^{\prime\prime})$, from which we can learn the ratio
$
{1}/({1-\zeta^{+}(x^{\prime\prime},x^\prime)}).
$
Together with \eqref{eq:H} this yields $H$. This concludes the proof of the corollary.
\end{proof}

The following example illustrates the usefulness of Corollary \ref{cor:onesided}.

\begin{example}[Stochastic volatility]
Consider a two-regime stochastic volatility model, which is a special case of Example \ref{example:HMM}. Assume that the outcome variable $Y$ has mean zero and conditional variance
$$
T \, \sigma_G^2 + (1-T) \, \sigma_H^2
$$ 
for positive constants $\sigma_G^2$ and $\sigma_H^2$. Suppose that $\sigma_G^2> \sigma_H^2$. Then $G$ is the distribution associated with a regime that is characterized by relatively higher volatility. In
this case, both tails of $G$ dominate those of $H$. Hence, in Assumption \ref{ass:dominance}, Condition (ii) holds but Condition (i) fails. Nevertheless, the distribution $H$ of the lower-volatility regime  remains identif\mbox{}ied.
\end{example}

Our identif\mbox{}ication result suggests plug-in estimators of the mixing proportions and the component distributions. 

The proof of Theorem \ref{thm:identification}, Equations
\eqref{eq:H}--\eqref{eq:G} in particular, further show that our mixture model yields overidentifying restrictions as soon as the instrument can take on more than two values. We turn to estimation
in the next section, where we also   construct a statistic for a specification test that  exploits the invariance of the formulae for $G$ and $H$ in Equations \eqref{eq:H}--\eqref{eq:G}   to the values $x^\prime,x^{\prime\prime}$.\footnote{The expression for $\lambda(x^\prime)$ in \eqref{eq:lambda} also holds for any $x^{\prime\prime}$. This invariance cannot fruitfully be exploited to test the tail restrictions of Assumption \ref{ass:dominance}, however, as the right-hand side expression in \eqref{eq:lambda} is independent of the value $x^{\prime\prime}$ even when Assumption \ref{ass:dominance} fails.}

\section{Estimation}
\label{sec:estimators}
To motivate the construction of our estimators, we f\mbox{}irst note that the structure of the model in \eqref{eq:model} continues to hold when we aggregate across $x$. Extending our notation to
$$
F(y\vert A) \equiv \mathbb{P}(Y\leq y \vert X\in A), \qquad
\lambda(A) \equiv \sum_{x\in A} \lambda(x) \ \mathbb{P}(X = x\vert X\in A),
$$
for any $A\subset \mathcal{X}$, we have 
\begin{equation} \label{eq:modelA}
F(y\vert A) = G(y) \, \lambda(A) + H(y) \, (1-\lambda(A)),
\end{equation}
which is of the same form as \eqref{eq:model}. Furthermore, the proof of Theorem \ref{thm:identification} continues to go through for \eqref{eq:modelA}; replacing $x^\prime$ with $A$ and $x^{\prime\prime}$ with $\mathcal{X}-A$ does not alter the argument. 

We will assume from now on that $X$ is discrete. As will become apparent, this only  entails a loss of generality for the estimation of the function $\lambda$, as our estimator will only yield a discretized approximation to it. Extending our results to continuous $X$ would complicate the exposition greatly and we feel that it would only distract from our main argument.

We will work under the following sampling condition.

\begin{assumption} \label{ass:sampling}
$(Y_1,X_1),\ldots,(Y_n,X_n)$ is a random sample on $(Y,X)$.
\end{assumption}

For each $A\subset\mathcal{X}$, let 
$$
F_n(y\vert A)\equiv n_A^{-1} \sum_{i=1}^n
\mathrm{1}\lbrace Y_i\leq y, X_i\in A \rbrace,
$$
where $n_A\equiv \sum_{i=1}^n \mathrm{1}\lbrace X_i \in A \rbrace$.

For each pair of disjoint subsets $A,B$ of $\mathcal{X}$ we can generalize \eqref{eq:zetas} to
\begin{equation} \label{eq:zetas2}
\begin{split}
\zeta^{-}(A,B)  \equiv \lim_{y\downarrow -\infty} \frac{F(y\vert A)}{F(y\vert B)} = \frac{1-\lambda(A)}{1-\lambda(B)},
\\
\qquad
\zeta^{+}(A,B)  \equiv \lim_{y\uparrow +\infty} \frac{1-F(y\vert A)}{1-F(y\vert B)} = \frac{\lambda(A)}{\lambda(B)}.
\end{split}
\end{equation}
For any subsample of size $m$ and integers $\iota_{m}$ and $\kappa_{m}$, let $\ell_{m}$ and $r_{m}$ denote the ($\iota_{m}+1$)th and ($m-\kappa_{m}$)th order statistics of $Y$ in this subsample. We estimate the quantities in \eqref{eq:zetas2} by
\begin{equation} \label{eq:xi}
\zeta_{n}^{-}(A,B) \equiv
\frac{F_n(\ell_{n_{B}}\vert A)}{F_n(\ell_{n_{B}}\vert B)},
\qquad
\zeta_{n}^{+}(A,B) \equiv
\frac{1-F_n(r_{n_{B}}\vert A)}{1-F_n(r_{n_{B}}\vert B)},
\end{equation}
respectively. In our asymptotic theory, we will choose $\iota_{n_B}$ and $\kappa_{n_B}$ so that $\ell_{n_{B}}\downarrow -\infty$ and $r_{n_{B}}\uparrow +\infty$ as $n\uparrow+\infty$ at an appropriate rate.

Estimators of both the mixing proportions and the component distributions follow readily along the lines of  the proof of Theorem~\ref{thm:identification}; see below. Since their asymptotic distribution will be driven by the large-sample behavior of the  estimators  of the quantities in \eqref{eq:xi}, we start by deriving the statistical properties of these estimators.

\subsection{Asymptotic theory for intermediate quantiles}

Throughout this section we f\mbox{}ix disjoint sets $A,B$ and consider the asymptotic behavior of the estimators in \eqref{eq:xi}.

Consistency only  requires the following rate conditions.

\begin{assumption}[Order statistics] \label{ass:orderstats1}
 $\iota_{n_{B}}/\sqrt{n_{B}\ln \ln n_{B}}\uparrow +\infty$ and $\kappa_{n_{B}}/\sqrt{n_{B}\ln \ln n_{B}}\uparrow +\infty$ as $n\uparrow+\infty$.
\end{assumption}

\begin{theorem} [Consistency] \label{thm:xi}
If Assumptions \ref{ass:model}--\ref{ass:orderstats1} hold, 
$$
\zeta_{n}^{-}(A,B) \overset{p}{\rightarrow} \zeta^{-}(A,B),
\qquad
\zeta_{n}^{+}(A,B)\overset{p}{\rightarrow}\zeta^{+}(A,B),
$$
as $n\uparrow+\infty$. 
\end{theorem}

\begin{proof}
We prove the theorem for $\zeta_n^+$; the proof for $\zeta_n^-$ follows in a similar fashion. Write
\begin{equation} \label{eq:debut1}
\zeta_{n}^{+}-\zeta^{+} =
(\zeta_{n}^{+}-\zeta^{\kappa_{n_B}}) + (\zeta^{\kappa_{n_B}}-\zeta^{+}),
\end{equation}
for $\zeta^{\kappa_{n_B}}\equiv (1-F(r_{n_B}\vert A))/(1-F(r_{n_B}\vert B))$. For the second right-hand side term in \eqref{eq:debut1} we have
\begin{equation*}
\begin{split}
\zeta^{\kappa_{n_B}}-\zeta^{+}
& =
\left(
\frac{\lambda(A) + \frac{1-H(r_{n_B})}{1-G(r_{n_B})} (1-\lambda(A))}{\lambda(B) + \frac{1-H(r_{n_B})}{1-G(r_{n_B})} (1-\lambda(B))}
-\frac{\lambda(A)}{\lambda(B)}\right)
 =
 \, O_p\left(\frac{1-H(r_{n_B})}{1-G(r_{n_B})}\right)
=
o_p(1),
\end{split}
\end{equation*}
by Assumptions \ref{ass:dominance}(ii) and \ref{ass:orderstats1}. To deal with the f\mbox{}irst right-hand side term in \eqref{eq:debut1}, recall that
$$
\zeta_{n}^{+}-\zeta^{\kappa_{n_B}}
=
\frac{1-F_n(r_{n_{B}}\vert A)}{1-F_n(r_{n_{B}}\vert B)}
-\frac{1-F(r_{n_{B}}\vert A)}{1-F(r_{n_{B}}\vert B)}.
$$
Letting $\mathbb{G}_{n}(y\vert S)\equiv \sqrt{n_S} \big(F_n(y\vert S)-F(y\vert S)\big)$ for any $S\subset\mathcal{X}$ we thus have that
\begin{align*}
\zeta_{n}^{+}-\zeta^{\kappa_{n_B}}
&=
\frac{(1-F(r_{n_B}\vert A))\mathbb{G}_n(r_{n_B}\vert B)/\sqrt{n_B}-(1-F(r_{n_B}\vert B))\mathbb{G}_n(r_{n_B}\vert A)/\sqrt{n_A}}{(1-F_n(r_{n_B}\vert B))(1-F(r_{n_B}\vert B))}
\\
&= \frac{\sqrt{n_B}}{\kappa_{n_B}}
\left(
\zeta^{\kappa_{n_B}} \mathbb{G}_{n}(r_{n_B}\vert B)- \sqrt{\frac{n_B}{n_A}}\mathbb{G}_{n}(r_{n_B}\vert A)
\right)\\
& = O_{a.s.}\left(\frac{\sqrt{n_B \, \ln \ln n_B}}{\kappa_{n_B}}\right),
\end{align*}
where the second equality uses $1-F_n(r_{n_B}\vert B)=\kappa_{n_B}/n_B$ and the last one follows by the law of the iterated logarithm for empirical processes. Thus, from Assumption \ref{ass:orderstats1} it follows that $\lvert\zeta_{n}^{+}-\zeta^{\kappa_{n_B}} \rvert = o_p(1)$. This completes the proof.
\end{proof}

Deriving the limit distribution requires some more care, and three more assumptions. We first  impose the following regularity condition on the component distributions.

\begin{assumption} \label{ass:continuity}
$G$ and $H$ are absolutely continuous on $\mathbb{R}$.
\end{assumption}

\noindent
This assumption is very weak. Note that, as we do not require the existence of moments of the component distributions, our results also apply to heavy-tailed distributions such as Cauchy and Pareto distributions.

We will complement Assumption \ref{ass:orderstats1} with an additional rate condition.

\begin{assumption}[Order statistics (cont'd.)] \label{ass:orderstats2} $\iota_{n_{B}}/n_{B}\downarrow 0$ and $\kappa_{n_{B}}/n_{B}\downarrow 0$ as $n\uparrow+\infty$.
\end{assumption}

\noindent
Where Assumption \ref{ass:orderstats1} required the order statistics to grow to ensure consistency, this assumption bounds this  growth rate  so that appropriately scaled versions of $\zeta_n^+$ and $\zeta_n^-$ have a limit distribution.

Finally, we will use an additional condition on the relative tails of the component distributions.

\begin{assumption}[Tail rates] \label{ass:tailrates} \mbox{} \\
(i) ${G(\ell_{n_{B}})}/{H(\ell_{n_{B}})} = o_p(1/\sqrt{\iota_{n_{B}}})$; and \\
(ii) $({1-H(r_{n_{B}})})/({1-G(r_{n_{B}})}) = o_p(1/\sqrt{\kappa_{n_{B}}})$.
\end{assumption}

\noindent
Assumption \ref{ass:tailrates} rules out distributions whose tails vanish too quickly and ensures that the limit distributions are free of asymptotic bias. We comment on Assumption \ref{ass:tailrates} after we derive the limit distributions of our estimators.

\medskip

Let $\rho_{A,B}\equiv\mathbb{P}(X\in B)/\mathbb{P}(X\in A)$. Note that $0<\rho_{A,B}<+\infty$ because of random sampling. Introduce
\begin{eqnarray*}
\sigma^2_{-}(A,B)
& \equiv & 
\zeta^{-}(A,B)^2 + \rho_{A,B} \, \zeta^{-}(A,B),
\\
\sigma^2_{+}(A,B)
& \equiv & 
\zeta^{+}(A,B)^2 + \rho_{A,B} \, \zeta^{+}(A,B),
\end{eqnarray*}
Theorem \ref{thm:xi} provides the asymptotic properties of the estimators in \eqref{eq:zetas2} and is the main building block for our subsequent results.

\begin{theorem} [Asymptotic normality] \label{thm:xi2}
If Assumptions \ref{ass:model}--\ref{ass:tailrates} hold, then as  $n\uparrow +\infty$,
\begin{eqnarray*}
\sqrt{\iota_{n_{B}}} 
\big(\zeta_{n}^{-}(A,B)-\zeta^{-}(A,B)\big)
& \overset{d}{\rightarrow} & \mathcal{N}(0,\sigma_{-}^2(A,B)),
\\
\sqrt{\kappa_{n_{B}}} 
\big(\zeta_{n}^{+}(A,B)-\zeta^{+}(A,B)\big)
& \overset{d}{\rightarrow} & \mathcal{N}(0,\sigma_{+}^2(A,B));
\end{eqnarray*}
and these two estimators are asymptotically independent.
\end{theorem}

\begin{proof}
We focus on the limit behavior of $\sqrt{\kappa_{n_B}}(\zeta_{n}^{+}-\zeta^{+})$ here; the proof of the result for $\sqrt{\iota_n}(\zeta_{n}^{-}-\zeta^{-})$ follows along similar lines.

As in the proof of Theorem \ref{thm:xi}, write
\begin{equation} \label{eq:debut}
\sqrt{\kappa_{n_B}}(\zeta_{n}^{+}-\zeta^{+}) =
\sqrt{\kappa_{n_B}}(\zeta_{n}^{+}-\zeta^{\kappa_{n_B}}) + \sqrt{\kappa_{n_B}}(\zeta^{\kappa_{n_B}}-\zeta^{+}),
\end{equation}
for $\zeta^{\kappa_{n_B}}\equiv (1-F(r_{n_B}\vert A))/(1-F(r_{n_B}\vert B))$.  Assumption \ref{ass:tailrates} implies that
\begin{equation*}
\begin{split}
\sqrt{\kappa_{n_B}}(\zeta^{\kappa_{n_B}}-\zeta^{+}) 
& =
\sqrt{\kappa_{n_B}} \, O_p\left(\frac{1-H(r_{n_B})}{1-G(r_{n_B})}\right) = o_p(1).
\end{split}
\end{equation*}
Hence, the second right-hand side term in \eqref{eq:debut} is asymptotically negligible. 

We now turn to the first term in \eqref{eq:debut}. From the proof of Theorem \ref{thm:xi} we have that
$$
\sqrt{\kappa_{n_B}}(\zeta_{n}^{+}-\zeta^{\kappa_{n_B}}) = 
\sqrt{\frac{n_{B}}{\kappa_{n_B}}}
\left(
\zeta^{\kappa_{n_B}} \mathbb{G}_{n}(r_{n_B}\vert B)- \sqrt{\frac{n_B}{n_A}}\mathbb{G}_{n}(r_{n_B}\vert A)
\right),
$$
where $\mathbb{G}_{n}(y\vert S)\equiv \sqrt{n_S} \big(F_n(y\vert S)-F(y\vert S)\big)$ for any $S\subset\mathcal{X}$. Let $\alpha_n(u)\equiv \sqrt{n}\big(\mathcal{U}_n(u)-u\big)$ for $\mathcal{U}_n$ the empirical cumulative distribution of an i.i.d.~sample of size $n$ from a uniform distribution on $[0,1]$. By Assumption \ref{ass:continuity},  $F(y\vert S)$ is continuous in $y$ for all $S\subset\mathcal{X}$. Therefore
$$
\mathbb{G}_{n}(y\vert A) = \alpha_{n_A}\big(1-F(y\vert A)\big)
\; \mbox{and } \;
\mathbb{G}_{n}(y\vert B) = \alpha_{n_B}\big(1-F(y\vert B)\big)
$$
by an application of the probability integral transform. Hence, we may write
\begin{equation} \label{eq:decomp}
\begin{split}
\sqrt{\kappa_{n_B}}(\zeta_{n}^{+}-\zeta^{\kappa_{n_B}})
& =
\zeta^{\kappa_{n_B}} \, \sqrt{\frac{n_{B}}{\kappa_{n_B}}}  \,
\alpha_{n_{B}}\big(1-F(r_{n_B}\vert B)\big)
\\ & -
\sqrt{\frac{n_{B}}{\kappa_{n_B}}} \sqrt{\frac{n_B}{n_A}}\,
\alpha_{n_{A}}\big(1-F(r_{n_B}\vert A)\big).
\end{split}
\end{equation}
We study the asymptotic behavior of each of the right-hand side terms in turn.

Start with the f\mbox{}irst  right-hand side term in \eqref{eq:decomp}. From the definition of the order statistic $r_{n_B}$,  we find by adding and subtracting $F_n(r_{n_B}\vert B)$ that
\[
1-F(r_{n_B}\vert B)=\frac{\kappa_{n_B}}{n_{B}}\left(1+ \frac{\sqrt{n_{B}}}{\kappa_{n_B}} \mathbb{G}_{n}(r_{n_B}\vert B)\right);
\]
or, defining $\varepsilon_n \equiv - \sqrt{n_{B}}/\kappa_{n_B} \, \mathbb{G}_{n}(r_{n_B}\vert B)$,
\[
1-F(r_{n_B}\vert B)=\frac{\kappa_{n_B}}{n_{B}}(1-\varepsilon_n).
\]
Therefore we can write
\begin{equation} \label{eq:firstterm}
\zeta^{\kappa_{n_B}} \, \sqrt{\frac{n_{B}}{\kappa_{n_B}}}  \,
\alpha_{n_{B}}\big(1-F(r_{n_B}\vert B)\big)
=
\sqrt{2}\, \zeta^{\kappa_{n_B}} \, \sqrt{\frac{n_{B}}{2\kappa_{n_B}}}   \,
\alpha_{n_{B}} \left( \frac{2\kappa_{n_B}}{n_{B}}\frac{1 - \varepsilon_n}{2}\right).
\end{equation}
By the law of the iterated logarithm together with Assumption~\ref{ass:orderstats1},
$$
\varepsilon_n = - \frac{\sqrt{n_B}}{\kappa_{n_B}} \, O_{a.s.}\left(\sqrt{\ln \ln n_B}\right)
=
O_{a.s.}\left(\frac{\sqrt{n_B \ln \ln n_B}}{\kappa_{n_B}}\right) = o_{a.s.}(1).
$$
Hence $(1-\varepsilon_n)/2$ converges almost surely to $1/2$; and $(1-\varepsilon_n)/2\in(0,1)$ for $n$ large enough. We may then apply Theorem 2.1 in \cite{Einmahl1992} to establish the convergence in distribution of 
$\sqrt{\frac{n_{B}}{2\kappa_{n_B}}}   \,
\alpha_{n_{B}} \left( \frac{2\kappa_{n_B}}{n_{B}}\frac{1 - \varepsilon_n}{2}\right)$ to a normal random variable with mean zero and variance $1/2$. This, together with Equation \eqref{eq:firstterm} and an application of Slutsky's theorem,  implies that
\begin{equation} \label{eq:term1}
\sqrt{\frac{n_{B}}{\kappa_{n_B}}} \zeta^{\kappa_{n_B}}  \,
\alpha_{n_{B}}\big(1-F(r_{n_B}\vert B)\big)
 \overset{d}{\rightarrow}
 \zeta^{+} Z^{+}_B,
\end{equation}
where $Z^{+}_B$ is a standard normal random variable.

Now turn to the second right-hand side term in \eqref{eq:decomp}. First observe that
$$
1-F(r_{n_B}\vert A)=\zeta^{\kappa_{n_B}}\big(1-F(r_{n_B}\vert B)\big)=
\zeta^{\kappa_{n_B}}\frac{\kappa_{n_B}}{n_B}(1-\varepsilon_n).
$$
Using  $\rho_{A,B}=\lim_{n\uparrow+\infty} n_B/n_A,$ this gives
\begin{equation*} 
\sqrt{\frac{n_{B}}{\kappa_{n_B}}} \sqrt{\frac{n_B}{n_A}}\,
\alpha_{n_{A}}\big(1-F(r_{n_B}\vert A)\big)
=
\sqrt{2\rho_{A,B} \, \zeta^{+}} \, \sqrt{\frac{n_A}{2\tilde\kappa_{n_A}} }\,
\alpha_{n_A}\left(
\frac{2\tilde\kappa_{n_A}}{n_A} \frac{1  - \varepsilon_n}{2}
\right)
+ o_p(1),
\end{equation*}
where $\tilde\kappa_{n_A}\equiv(\kappa_{n_B}\zeta^{\kappa_{n_B}})/(n_B/n_A)$. As $\tilde\kappa_{n_A}$  satisf\mbox{}ies Assumption~(1.5) of Theorem 2.1 in \cite{Einmahl1992} we may apply his theorem again to obtain
\begin{equation} \label{eq:term2}
\sqrt{\frac{n_{B}}{\kappa_{n_B}}} \sqrt{\frac{n_B}{n_A}}\,
\alpha_{n_{A}}\big(1-F(r_{n_B}\vert A)\big)
\overset{d}{\rightarrow} \sqrt{\rho_{A,B}\, \zeta^{+}}\, Z^{+}_A,
\end{equation}
where $Z^{+}_A$ is a standard-normal random variable which, because of random sampling, is independent of $Z^{+}_B$.

Combining \eqref{eq:decomp} with \eqref{eq:term1} and \eqref{eq:term2} then gives
$$
\sqrt{\kappa_{n_B}}(\zeta_{n}^{+}-\zeta^{\kappa_{n_B}}) \overset{d}{\rightarrow} 
\zeta^{+} Z^{+}_B - \sqrt{\rho_{A,B} \, \zeta^{+}} \, Z^{+}_A, 
$$
as claimed. This concludes the proof.
\end{proof}

We finish this section with two examples that specialize Assumption \ref{ass:tailrates} to densities with log-concave tails and Pareto tails, respectively. In both cases, Assumption~\ref{ass:tailrates} is implied by the rate conditions in Assumption \ref{ass:orderstats2}.

\begin{example}[Log-concave tails] \label{example:concavetails}
Suppose that $G$ and $H$ have log-concave tails; and for notational simplicity, assume that 

$$
-\ln\left(1-G(y)\right) \sim \left( \frac{y}{\sigma_G^+}\right)^{\alpha^+_G}, \quad
-\ln\left(1-H(y)\right) \sim \left( \frac{y}{\sigma_H^+}\right)^{\alpha^+_H},  \quad \text{as } y\uparrow +\infty,
$$
for real numbers $\alpha^+_G,\alpha^+_H>1$ and $\sigma^+_G,\sigma^+_H>0$, and  
$$
-\ln G(y) \sim \left( \frac{-y}{\sigma_G^-}\right)^{\alpha^-_G}, \quad
-\ln H(y) \sim \left( \frac{-y}{\sigma_H^-}\right)^{\alpha^-_H},  \quad \text{as } y\downarrow -\infty,
$$
for real numbers $\alpha^-_G,\alpha^-_H>1$ and $\sigma^-_G,\sigma^-_H>0$. Then Assumption~\ref{ass:orderstats2} implies  Assumption~\ref{ass:tailrates}  if both

\vspace{.25cm}\noindent (i)
$\alpha_G^+ < \alpha_H^+$,  \emph{or\/} $\alpha_G^+ = \alpha_H^+$ and
  $\sigma^+_G>\sigma^+_H$; and
  
\vspace{.25cm}\noindent (ii)
$\alpha_G^- > \alpha_H^-$, \emph{or\/} $\alpha_G^- = \alpha_H^-$ and $\sigma^-_G<\sigma^-_H$

\vspace{.25cm} \noindent
hold.
\end{example} 

\begin{proof} 
We verify the second rate; the f\mbox{}irst follows similarly. Throughout, f\mbox{}ix the set $B$. Assumptions \ref{ass:dominance}(ii) and \ref{ass:orderstats2} imply that
\begin{equation*}
\begin{split}
1-F(r_{n_B}\vert B) 
& = 
\left(1-G(r_{n_B})\right) \, \lambda(B) + \left(1-H(r_{n_B})\right) \, \left(1-\lambda(B)\right) 
\\
& =
\left(1-G(r_{n_B})\right) \, \left(\lambda(B) + o_p(1)\right).
\end{split}
\end{equation*}
Further, because $\kappa_{n_B}/n_B=1-F_n(r_{n_B}\vert B)$, adding and subtracting $F(r_{n_B}\vert B)$ gives
\begin{equation*}
\begin{split}
\frac{\kappa_{n_B}}{n_B} 
& = 
\left(1-F(r_{n_B}\vert B)\right) + \left(F_n(r_{n_B}\vert B)-F(r_{n_B}\vert B)\right)
\\
& =
\left(1-F(r_{n_B}\vert B)\right) + O_{a.s.} (\sqrt{(\ln \ln n_B)/n_B}).
\end{split}
\end{equation*}
Because $(\ln \ln n_B)/n_B\rightarrow 0$, put together, we f\mbox{}ind
$$
\frac{\kappa_{n_B}}{n_B}  = C\, \left(1-G(r_{n_B})\right) \, (1+ o_p(1)) 
$$
for some constant $C$. Since $G$ and $H$ have log-concave tails, it  follows from this expression that $r_{n_B}$ behaves asymptotically like $\sqrt[^{\alpha_G^+}]{\ln n_B}$. And since
$$
\frac{1-H( r_{n_B})}{1-G( r_{n_B})}  \sim \exp \left\lbrace  \left(\frac{r_{n_B}}{\sigma_G^+}\right)^{\alpha_G^+}-  \left(\frac{r_{n_B}}{\sigma_H^+}\right)^{\alpha_H^+} \right\rbrace,
$$
we have that
\begin{equation*}
\begin{split}
\frac{1-H( r_{n_B})}{1-G( r_{n_B})}  =
\left\lbrace\begin{array}{ll}
O_p\left(\exp (-  (\ln n_B)^{\alpha_H^+/\alpha_G^+})\right) 
& \text{if } \alpha_H^+ > \alpha_G^+ \\
O_p\left(1/n_B\right) & \text{if } \alpha_H^+ = \alpha_G^+ \text{ and } \sigma_H^+ < \sigma_G^+ 
\end{array}\right.   ,
\end{split}
\end{equation*}
from which the conclusion follows.
\end{proof}

\noindent
Example \ref{example:concavetails} does not cover  location models with log-concave distributions in the  case when  the $\alpha$ and $\sigma$ parameters of $H$ equal those of $G$. This includes the location model with Gaussian errors, for which  $\alpha=2$ and $\sigma$ is the common standard error. While our estimator remains consistent in such cases, we do not know of general results on tail empirical processes that would yield the asymptotic distribution of the estimator in this knife-edge case. To assess the extent to which the failure of Assumption \ref{ass:tailrates} may play a role for inference, our simulation experiments in Section~\ref{sec:MC} include a Gaussian location model.

\begin{example}[Pareto tails] \label{example:paretotails}
Let $C$ denote a generic constant. Suppose that $G$ and $H$ have Pareto tails, i.e., 
$$
\left(1-G(y)\right) \sim C \, y^{-\alpha_G^+}, \quad
\left(1-H(y)\right) \sim C \, y^{-\alpha_H^+},  \quad \text{as } y\uparrow +\infty,
$$
for positive real numbers $\alpha^+_H> \alpha^+_G$ and
$$
G(y) \sim C \, (-y)^{-\alpha_G^-}, \quad
H(y) \sim C \, (-y)^{-\alpha_H^-},  \quad \text{as } y\downarrow -\infty,
$$
for positive real numbers $\alpha^-_G< \alpha^-_H$. 
Then Assumption \ref{ass:orderstats2} implies Assumption \ref{ass:tailrates}.
\end{example}

\begin{proof} 
The argument is very similar to the one that was used to verify Example \ref{example:concavetails}.
We focus on  the right tail; the argument for the left tail is similar. We have
$$
\frac{\kappa_{n_B}}{n_B} = \left(1- G( r_{n_B}) \right) \left( 1+o_p(1)\right) = C \, r^{-\alpha_G^+} \, \left(1+o_p(1)\right).
$$
Assumption \ref{ass:tailrates} requires that $(1-H( r_{n_B}))/(1-G( r_{n_B}))=o(1/\sqrt{\kappa_n})$, that is, that $r_{n_B}^{\alpha_G^+-\alpha_H^+} = o_p(1/\sqrt{\kappa_{n_B}})$. This rate condition is satisf\mbox{}ied when
$$
\left(\frac{{n_B}}{\kappa_{n_B}}\right)^{\frac{\alpha_G^+-\alpha_H^+}{\alpha_G^+}} = 
o_p\left(\frac{1}{\sqrt{\kappa_{n_B}}}\right),
$$
which can be achieved by setting $\kappa_{n_B} = o(n_B^{\gamma^+})$ for 
\begin{equation} \label{eq:paretocoeff}
\gamma^+\equiv 
\frac{\alpha^+_H-\alpha_G^+}{\alpha^+_H-\alpha_G^+/2}.
\end{equation}
This condition is weaker than Assumption \ref{ass:orderstats2} and is therefore implied by it.
\end{proof}

\noindent
Example \ref{example:paretotails} shows that our methods are well suited to deal with Pareto tails.
Pareto tails show up in many economic applications. A time-honored example is  income and wealth distributions (\citealt{AtkinsonPikettySaez2011}), which are often modeled as a log-normal for most quantiles, combined with a Pareto right tail. More generally, ``power laws'' have become a popular tool in finance, in studies of firm growth, and in urban economics (see \citealt{Gabaix2009} for a recent survey, and \citealt{AcemogluEtAl2012} for an application to business cycles.)
Many recent models of monopolistic competition, as used in international trade for instance, also assume that  productivities are Pareto-distributed (\citealt{ArkolakisCostinotRodriguezClare2012}).

Let us focus on the right tail condition. Identif\mbox{}ication  only requires that the tail index of $H$ be larger than that of  $G$, that is, $\alpha_H^+ >\alpha_G^+$. Let  $c^+\equiv\alpha_H^+/\alpha_G^+>1$. Equation \eqref{eq:paretocoeff} then gives a convergence rate arbitrarily close to $n^{-\beta^+/2}$ for $\beta^+ = 2(c^+-1)/(2c^+-1)$. For example, if $c^+=2$ then $\beta^+=2/3$ and our estimators will converge slightly slower than $n^{-1/3}$. However, as $c^+$ increases, $\beta^+$ becomes closer to one and our estimators will converge at close to the $n^{-1/2}$ parametric rate.

\subsection{Mixing proportions}
Fix $x\in\mathcal{X}$ and consider estimating $\lambda(x)$. Set $A=\mathcal{X}- x$ and $B= x$ in \eqref{eq:zetas2} and solve for $\lambda(x)$ to get 
$$
\lambda(x) = \frac{1-\zeta^{-}(A,x)}{\zeta^{+}(A,x)-\zeta^{-}(A,x)}.
$$
The mixing proportion $\lambda$ need not be a strictly monotonic function. Estimating $\lambda(x)$ by an average of plug-in estimates of \eqref{eq:lambda} could therefore be problematic, as the denominator in \eqref{eq:lambda} can be zero or be arbitrarily close to it for some pairs of values $(x^\prime,x^{\prime\prime})$.

We instead estimate the mixing proportion at $X=x$ by  a plug-in estimator based on \eqref{eq:xi}, that is,
$$
\lambda_n(x) \equiv \frac{1-\zeta_n^{-}(A,x)}{\zeta_n^{+}(A,x)-\zeta_n^{-}(A,x)}.
$$
This estimator uses observations with $X_i \neq x$ in a way that immunizes it against small or zero denominators.

To present the asymptotic variance of this estimator we need to def\mbox{}ine
\begin{equation} \label{eq:jacobian}
\begin{split}
d^{-}(x) 
& 
\equiv 
\frac{1-\zeta^{+}(A,x)}{( \zeta^{+}(A,x)-\zeta^{-}(A,x) )^2} ,
\\
d^{+}(x) 
&  
\equiv 
\frac{\zeta^{-}(A,x)-1}{( \zeta^{+}(A,x)-\zeta^{-}(A,x) )^2}.
\end{split}
\end{equation}
The speed of convergence and the asymptotic distribution of the $\lambda_n(x)$ depend on the ratio $c_x\equiv \lim_{n\uparrow +\infty} {\iota_{n_x}/\kappa_{n_x}}$.

\begin{theorem}[Mixing proportions] \label{thm:proportions}
Under the conditions of Theorem \ref{thm:xi}, 
$$
\lvert \lambda_n(x) - \lambda(x) \rvert = o_p(1)
$$
as $n\uparrow+\infty$. 

\noindent
Under the conditions of Theorem \ref{thm:xi2},
\begin{eqnarray*}
&\sqrt{\iota_{n_x}} \big(\lambda_n(x)-\lambda(x)\big) 
 \overset{d}{\rightarrow}  
\mathcal{N}\big(0 \, , \, d^{-}(x)^2 \sigma_{-}^2(A,x) + c_x \, d^{+}(x)^2 \sigma_{+}^2(A,x)\big) 
& \text{if } c_x <+\infty,
\\
&\sqrt{\kappa_{n_x}} \big(\lambda_n(x)-\lambda(x)\big) 
 \overset{d}{\rightarrow} 
\mathcal{N}\big(0 \, , \, c_x^{-1} d^{-}(x)^2 \sigma_{-}^2(A,x) +  d^{+}(x)^2 \sigma_{+}^2(A,x)\big)
& \text{if } c_x > 0,
\end{eqnarray*}
as $n\uparrow+\infty$.
\end{theorem}

\begin{proof}
The consistency claim follows directly from Theorem \ref{thm:xi} by an application of the continuous mapping theorem.

To establish the asymptotic distribution, note that Theorem~\ref{thm:xi2} states that
\begin{eqnarray*}
\sqrt{\iota_{n_x}}  (\zeta_n^-(A,x)-\zeta^-(A,x)) & \overset{d}{\rightarrow} & \mathcal{N}(0,\sigma_{-}^2(A,x)), \\
\sqrt{\kappa_{n_x}} (\zeta_n^+(A,x)-\zeta^+(A,x)) & \overset{d}{\rightarrow} & \mathcal{N}(0,\sigma_{+}^2(A,x)),
\end{eqnarray*}
and that $\zeta_n^-(x)$ and $\zeta_n^+(x)$ are asymptotically independent. An expansion around $\zeta^-(A,x)$ and $\zeta^+(A,x)$ then yields
\begin{equation*}
\begin{split}
\sqrt{\iota_{n_x}} (\lambda_n(x) - \lambda(x)) 
& = d^-(x) \, \sqrt{\iota_{n_x}}  (\zeta_n^-(A,x)-\zeta^-(A,x))
\\  
& + d^+(x) \, \sqrt{\kappa_{n_x}} (\zeta_n^+(A,x)-\zeta^+(A,x)) \, \sqrt{\frac{\iota_{n_x}}{\kappa_{n_x}}} + o_p(1),
\end{split}
\end{equation*}
which has the limit distribution stated in the theorem if $c_x$ is f\mbox{}inite. Also, by the same argument, 
\begin{equation*}
\begin{split}
\sqrt{\kappa_{n_x}} (\lambda_n(x) - \lambda(x)) 
& = d^+(x) \, \sqrt{\kappa_{n_x}}  (\zeta_n^+(x)-\zeta^+(x))
\\  
& + d^-(x) \, \sqrt{\iota_{n_x}} (\zeta_n^-(x)-\zeta^-(x)) \, \sqrt{\frac{\kappa_{n_x}}{\iota_{n_x}}} + o_p(1)
\end{split}
\end{equation*}
converges in distribution as stated in the theorem if $c_x$ is non-zero. This verif\mbox{}ies the claims and proves the theorem.
\end{proof}




\subsection{Component distributions}

To estimate the component distributions, choose $B=\mathcal{X}-A$ so that $A$ and $B$   partition $\mathcal{X}$. Equations \eqref{eq:H} and \eqref{eq:G} then suggest the estimators
\begin{equation} \label{eq:estGH}
\begin{split}
H_n(y;A,B) 
& \equiv 
F_n(y\vert A) - \frac{1}{1-\zeta_n^{+}(B,A)} \, \big( F_n(y\vert A)-F_n(y\vert B) \big),
\\
G_n(y;A,B) 
& \equiv 
F_n(y\vert A) - \frac{1}{1-\zeta_n^{-}(B,A)} \, \big( F_n(y\vert A)-F_n(y\vert B) \big) .
\end{split}
\end{equation}
For notational simplicity we now drop $A$ and $B$ from the arguments: $G_n(y) \equiv G_n(y;A,B)$ and $H_n(y)\equiv H_n(y;A,B)$.

To state their asymptotic behavior, let
\begin{eqnarray*}
d_{G}(A,B;y)
& \equiv &
\frac{F(y\vert A)-F(y\vert B)}{(1-\zeta^-(B,A))^2}, 
\\
d_{H}(A,B;y)
& \equiv &
\frac{F(y\vert A)-F(y\vert B)}{ (1-\zeta^+(B,A))^2}, 
\end{eqnarray*}
and let $\lVert\cdot \rVert_{\infty}$ denote the supremum norm.

\begin{theorem} \label{thm:distributions}
Under the conditions of Theorem \ref{thm:xi}, 
$$
\lVert G_n-G \rVert_\infty = o_p(1), \qquad
\lVert H_n-H \rVert_\infty = o_p(1), 
$$
as $n\uparrow +\infty$. 

\noindent
Under the conditions of Theorem \ref{thm:xi2}, 
\begin{eqnarray*}
\sqrt{\iota_{n_A}} (G_n(y)-G(y))  & \overset{d}{\rightarrow} & 
\mathcal{N}\left(0,d_{G}(A,B;y)^2 \, \sigma_{-}^2(B,A)\right),
\\
\sqrt{\kappa_{n_A}} (H_n(y)-H(y)) & \overset{d}{\rightarrow} & 
\mathcal{N}\left(0, d_{H}(A,B;y)^2 \, \sigma_{+}^2(B,A)\right),
\end{eqnarray*}
as $n\uparrow+\infty$ for each $y\in\mathbb{R}$,.
\end{theorem}

\begin{proof}
Consistency follows by Theorem \ref{thm:xi} and the Glivenko-Cantelli theorem.

We establish the asymptotic distribution of $G_n$; the result for $H_n$ follows by the same argument.

First note that
$$
\sqrt{\iota_{n_A}} (G_n(y)-G(y))  = T_1 + T_2 + T_3
$$
for
\begin{eqnarray*}
T_1  & \equiv &
\sqrt{\iota_{n_A}} (F_n(y\vert A)-F(y\vert A)),
\\
T_2 & \equiv &
-\frac{1}{1-\zeta^-(B,A)} \,
\sqrt{\iota_{n_A}} \, 
\left( 
\big\lbrace F_n(y\vert A) - F(y\vert A) \big\rbrace - \big\lbrace F_n(y\vert B) - F(y\vert B) \big\rbrace
\right),
\\
T_3 & \equiv &
- (F_n(y\vert A)-F_n(y\vert B)) \, \sqrt{\iota_{n_A}} \left(\frac{1}{1-\zeta_n^-(B,A)}-\frac{1}{1-\zeta^-(B,A)}\right).
\end{eqnarray*}
By the Glivenko-Cantelli theorem, $T_1=o_p(1)$ and $T_2=o_p(1)$ while
$$
T_3 = 
- (F(y\vert A)-F(y\vert B)) \, \sqrt{\iota_{n_A}} \left(\frac{1}{1-\zeta_n^-(B,A)}-\frac{1}{1-\zeta^-(B,A)}\right)+ o_p(1).
$$
A linearization of this expression in $\zeta_n^{-}(B,A)-\zeta^-(B,A)$ together with an application of Theorem \ref{thm:xi2} to the partition $A,B$ then yields the result.
\end{proof}

When $X$ can take on more than two values there are multiple ways of choosing the sets $A$ and $B$. Inspection of the asymptotic variance does not give clear guidance on how to choose $A$ and $B$ in an optimal manner. An ad-hoc way to proceed when the number of possible choices for $A,B$ is small, is to simply compute estimators for all possible choices. Alternatively, it would be possible to combine estimates based on multiple choices through a minimum-distance procedure. We leave a detailed analysis for future research.

\subsection{Specif\mbox{}ication testing}


An implication of our model restrictions is that the estimators of $G$ and $H$ in \eqref{eq:estGH}, when based on dif\mbox{}ferent subsets of $\mathcal{X}$,  should co-incide with one another, up to sampling error. This observation suggests the possibility to test the specif\mbox{}ication when $X$ can take on more than two values.

Theorem \ref{thm:exclusion} provides the relevant asymptotic distributional result to perform this test. In it we use
\begin{equation*}
\begin{split}
\Sigma_G
& =
d_G(A,C) \left\lbrace d_G(A,C) \, \sigma_{-}^2(C,A) - d_G(A,B) \, \zeta^{-}(C,A) \zeta^{-}(B,A) \right\rbrace
\\
& +
d_G(A,B) \left\lbrace d_G(A,B) \, \sigma_{-}^2(B,A) - d_G(A,C) \, \zeta^{-}(C,A) \zeta^{-}(B,A) \right\rbrace
\end{split}
\end{equation*}
and
\begin{equation*}
\begin{split}
\Sigma_H
& =
d_H(A,C) \left\lbrace d_H(A,C) \, \sigma_{+}^2(C,A) - d_H(A,B) \, \zeta^{+}(C,A) \zeta^{+}(B,A) \right\rbrace
\\
& +
d_H(A,B) \left\lbrace d_H(A,B) \, \sigma_{+}^2(B,A) - d_H(A,C) \, \zeta^{+}(C,A) \zeta^{+}(B,A) \right\rbrace ,
\end{split}
\end{equation*}
where the triple $A,B,C$ constitutes any partition of $\mathcal{X}$ and, for any $A$ and $B$, we write
$$
d_G(A,B) \equiv \mathbb{E}[W(Y) d_G(A,B;Y)], \qquad
d_H(A,B) \equiv \mathbb{E}[W(Y) d_H(A,B;Y)],
$$
for a chosen weight function $W$ that is bounded on $\mathbb{R}$. The  choice of these weights should reflect the analyst's concerns about potential violations of our assumptions in the application under study.

\begin{theorem}[Specif\mbox{}ication testing] \label{thm:exclusion}
Under the conditions of Theorem \ref{thm:xi2}
$$
\lim_{n\uparrow +\infty} \mathbb{P}
\left\lbrace \left\lvert
\frac{n^{-1}\sum_{i=1}^n W(Y_i)G_n(Y_i; A,B) - n^{-1}\sum_{i=1}^nW(Y_i)G_n(Y_i; A,C)}{\sqrt{\Sigma_G } / \sqrt{\iota_{n_A}}}
\right\rvert
 > z(\tau/2)
\right\rbrace = \tau ,
$$
and
$$
\lim_{n\uparrow +\infty} \mathbb{P}
\left\lbrace \left\lvert
\frac{n^{-1}\sum_{i=1}^nW(Y_i)H_n(Y_i; A,B) - n^{-1}\sum_{i=1}^nW(Y_i)H_n(Y_i; A,C)}{\sqrt{\Sigma_H } / \sqrt{\kappa_{n_A}}}
\right\rvert
 > z(\tau/2)
\right\rbrace = \tau ,
$$
where $z(\tau)$ is the $1-\tau$ quantile of the standard-normal distribution.
\end{theorem}

\begin{proof}
We consider only the case of $G$. The dif\mbox{}ference $G_n(y; A,B) - G_n(y; A,C)$ equals
$$
\frac{1}{1-\zeta^{-}_n(C,A)}
(F_n(y\vert A) - F_n(y\vert C))
-
\frac{1}{1-\zeta^{-}_n(B,A)}
(F_n(y\vert A) - F_n(y\vert B))
$$
for any $y$. An expansion around $\zeta^{-}(C,A)$ and $\zeta^{-}(B,A)$, then shows that the scaled dif\mbox{}ference $\sqrt{\iota_{n_A}}G_n(y; A,B) - G_n(y; A,C)$ is asymptotically equivalent to
$$
d_G(A,C;y) \, \sqrt{\iota_{n_A}} \, \big(\zeta_n^-(C,A)-\zeta^-(C,A)\big)
-
d_G(A,B;y) \, \sqrt{\iota_{n_A}} \, \big(\zeta_n^-(B,A)-\zeta^-(B,A)\big).
$$
This holds for any $y$ and, therefore, also for the weighted average over $y$. Together with Theorem \ref{thm:xi2}, this result then readily yields the asymptotic distribution of the dif\mbox{}ference $n^{-1}\sum_{i=1}^n W(Y_i)G_n(Y_i; A,B) - n^{-1}\sum_{i=1}^nW(Y_i)G_n(Y_i; A,C)$ and implies the claim of the theorem.
\end{proof}

We leave a detailed analysis of the power properties of this specif\mbox{}ication test for future research. Here, we provide a consistency result against failure of Assumption \ref{ass:dominance}.

\begin{example}[Consistency of the test]
Suppose that $H$ dominates $G$ in both tails. Then $H$ is no longer identif\mbox{}ied and
$$
\lim_{n\uparrow +\infty} 
\mathbb{P}
\left\lbrace \left\lvert
\frac{n^{-1}\sum_{i=1}^nW(Y_i)H_n(Y_i; A,B) - n^{-1}\sum_{i=1}^nW(Y_i)H_n(Y_i; A,C)}{\sqrt{\Sigma_H } / \sqrt{\kappa_{n_A}}}
\right\rvert
 > z
\right\rbrace = 1
$$
for any $z$.
\end{example}

\begin{proof}
When $H$ dominates $G$ in both tails, a small calculation reveals that
$$
\zeta_n^+(A,B) = \zeta^-(A,B) + o_p(1),
$$
and so  
$
\sqrt{\kappa_{n_A}} \, \lvert (\zeta_n^+(A,B) -\zeta^+(A,B) ) \rvert
$
grows without bound as $n\uparrow +\infty$. The conclusion then readily follows from the linearization in the proof of Theorem \ref{thm:exclusion}.
\end{proof}

\section{Simulation experiments}
\label{sec:MC}

In our numerical illustrations we will work with the family of skew-normal distributions (\citealt{Azzalini1985}). The skew-normal distribution  with location $\mu$, positive scale $\sigma$, and skewness parameter $\beta$ multiplies the density of $\mathcal{N}(\mu,\sigma^2)$ by a term that skews it to the right if $\beta>0$ and to the left if $\beta<0$:
$$
f(x;\mu,\sigma,\beta) \equiv
\frac{1}{\sigma} \, \phi\left(\frac{x-\mu}{\sigma}\right) \times \frac{\Phi\left(\beta \, \frac{x-\mu}{\sigma}\right)}{\Phi(0)} .
$$
Its  mean and variance are $\mu+\sigma\delta \sqrt{\frac{2}{\pi}}$ and $\sigma^2 \left(1-\frac{2\delta^2}{\pi}\right)$, respectively, where $\delta \equiv \beta/\sqrt{1+\beta^2}$.  Clearly, 
$$
f(x;\mu,\sigma,\beta) \rightarrow \frac{1}{\sigma} \, \phi\left(\frac{x-\mu}{\sigma}\right)
$$
as $\beta\rightarrow 0$.

In our simulations we will consider data generating processes where the outcome is generated as
\begin{equation} \label{eq:mc}
Y = T \, V_G + (1-T) \, V_H,
\end{equation}
where $T$ is a latent binary variable, and $V_G\sim G$ and $V_H\sim H$. Both error distributions $G$ and $H$ are skewed-normal distributions with parameters $\mu_G,\sigma_G,\beta_G$ and $\mu_H,\sigma_H,\beta_H$, respectively.

From \cite{Capitanio2010} it follows that Assumption \ref{ass:tailrates} holds if $G$ is right-skewed and $H$ is left-skewed. We will consider designs where $\beta_G>0$ and $\beta_H<0$ to verify our asymptotics. 

When $\beta_G=\beta_H=0$, \eqref{eq:mc} collapses to a standard location model with normal errors
\begin{equation} \label{eq:mc2}
Y =  (\mu_G-\mu_H) \, T + \, V, \qquad V \sim \mathcal{N}(0,\sigma_G^2+\sigma_H^2).
\end{equation}
The identifying tail condition in Assumption \ref{ass:dominance} still holds if $\mu_G> \mu_H$, and our estimators remain consistent. However, Assumption \ref{ass:tailrates} now fails and so we may expect poor inference in this design.

In our experiments we generate a binary $X$ with $\mathbb{P}(X=1)=\frac{1}{2}$ and f\mbox{}ix conditional probabilities as
\begin{eqnarray*}
 \mathbb{P}(T=0\vert X=0)  = \frac{3}{4}, && \mathbb{P}(T=1\vert X=0) = \frac{1}{4}, \\
 \mathbb{P}(T=1\vert X=1)  = \frac{1}{4}, && \mathbb{P}(T=1\vert X=1) = \frac{3}{4}. 
\end{eqnarray*}
We present results for data generating processes where $\mu_G=\mu=-\mu_H$ and $\beta_G=\beta=-\beta_H$. We use the designs $\mu=0 $ and $\beta\in \lbrace 2.5, 5 \rbrace$ to evaluate the adequacy of our asymptotic arguments for small-sample inference. We also look at the performance of our estimators when $\mu\in\lbrace .5, 1 \rbrace$ and $\beta=0$, which yields the Gaussian  location model in~\eqref{eq:mc2}. We f\mbox{}ix $\sigma_G=\sigma_H=1$ throughout. For each of these designs we consider choices of the empirical quantiles as
$$
\iota_{n_x} = C \, (n_x \ln \ln n_x)^{6/10}, \qquad
\kappa_{n_x} = C \, (n_x \ln \ln n_x)^{6/10}, 
$$
for several choices of the constant $C$. All of these choices are in line with our asymptotic arguments. The larger the constant $C$ the more conservative the choice of intermediate quantile,
$$
q_\ell \equiv \frac{\iota_{n_x}}{n_x}, \qquad q_r \equiv \frac{n_x-\kappa_{n_x}}{n_x},
$$
for a given sample size. 

We run experiments for sample sizes $n\in \lbrace 500; 1,000, 2,500; 5,000; 10,000; 25,000 \rbrace.$ We report (the average over the replications of) $q_\ell$ and $q_r$ along with the estimation results to get an idea of how far in the tails of the component distributions we are going to obtain the results. A data-driven determination of the constant $C$ is challenging and is left for future research.  For space considerations we report only a subset of the results here. The full set of simulation results is available in the working paper version of this paper \citep{JochmansHenrySalanie2014}.

Tables \ref{table:design1full} and \ref{table:remaining} report the results for the mixing proportions $\lambda(0)$ and $\lambda(1)$. Each table contains the bias, standard deviation (SD), ratio of the (average over the replications of the) estimated standard error to the standard deviation (SE/SD), and the coverage of $95\%$ conf\mbox{}idence intervals (CI95) for $n\in \lbrace 1,000, 10,000 \rbrace$. All these statistics were computed from $10,000$ Monte Carlo replications. Table \ref{table:design1full} reports results for the simulation design with $\mu=0,\beta=5$ for $C\in\lbrace .5, 1, 1.5  \rbrace$, so as to evaluate the impact of the choice of this tuning parameter on the results. This impact was similar in all other designs and so, for these designs, we present only results for one choice of $C$. The constant $C$ was f\mbox{}ixed to $.5$ for all designs except for the pure location model with $\mu=.5$ and $\beta=0$, where, for practical reasons, we use $C=.75$.\footnote{In this design, there is a small probability that either $q_\ell=0$ or $q_{r}=1$ when $C=.5$ and $n$ is small. This shows up in simulations with a large number of replications, as is the case here. The slightly more conservative choice of $C=.75$ avoids this issue.} These results are bundled in Table \ref{table:remaining}.

\begin{table}[htbp]
\caption{Mixing proportions}
\begin{center}
\begin{tabular}{rrrrrrrrrrr}

\hline\hline
\multicolumn{3}{l}{} & \multicolumn{2}{c}{BIAS}  & \multicolumn{2}{c}{SD} & \multicolumn{2}{c}{SE/SD} & \multicolumn{2}{c}{CI$95$}  \\ 
\multicolumn{1}{c}{$n$} & \multicolumn{1}{c}{$q_\ell$} & \multicolumn{1}{c}{$q_r$} & \multicolumn{1}{l}{$\lambda_n(0)$} & \multicolumn{1}{l}{$\lambda_n(1)$} & \multicolumn{1}{l}{$\lambda_n(0)$} & \multicolumn{1}{l}{$\lambda_n(1)$} & \multicolumn{1}{l}{$\lambda_n(0)$} & \multicolumn{1}{l}{$\lambda_n(1)$} & \multicolumn{1}{l}{$\lambda_n(0)$} & \multicolumn{1}{l}{$\lambda_n(1)$} \\ 
\hline
\multicolumn{11}{c}{$C=.5$}  \\

$ 1,000$ & $.059$ & $.940$ & $.0060$ & $-.0059$ & $.0693$ & $.0701$ & $1.0554$ & $1.0392$&$.9688$&$.9682$ \\ 
$10,000$ & $.026$ & $.974$ & $.0012$ & $-.0011$ & $.0328$ & $.0325$ & $1.0106$ & $1.0213$&$.9560$&$.9572$ \\ 

\hline
\multicolumn{11}{c}{$C=1$}  \\ 

$ 1,000$ & $.120$ & $.880$ & $.0024$ & $-.0035$ & $.0439$ & $.0446$ & $1.1358$ &$1.1220$ &$.9764$&$.9752$ \\ 
$10,000$ & $.052$ & $.947$ & $.0007$ & $-.0003$ & $.0225$ & $.0222$ & $1.0360$ &$1.0519$ &$.9566$&$.9616$ \\ 

\hline
\multicolumn{11}{c}{$C=1.5$}  \\ 

$ 1,000$ & $.179$ & $.821$ & $.0046$ & $-.0037$ & $.0316$ & $.0315$ & $1.2931$ & $1.2933$ &$.9944$&$.9920$\\ 
$10,000$ & $.078$ & $.922$ & $.0002$ & $-.0010$ & $.0175$ & $.0174$ & $1.0873$ & $1.0962$ &$.9646$&$.9710$\\ 

\hline\hline

\end{tabular}
\end{center}
\label{table:design1full}
\end{table}

The results in Table \ref{table:design1full} support our asymptotic theory. For  all choices of the tuning parameter $C$, the bias and standard deviation shrink to zero as $n\uparrow+\infty$; and the bias is small relative to the standard error. Furthermore, $\mathrm{SE}/\mathrm{SD}\rightarrow 1$ and the coverage rates of the conf\mbox{}idence intervals are close to $.95$ in large samples.  The variability of the point estimates is somewhat overestimated when $n$ is very small and $C$ is chosen conservatively. Together with the relatively small bias, this implies that conf\mbox{}idence intervals are slightly conservative. For $C=.5$, coverage rates are close to $.95$, even for the smallest samples considered, and for all $C$, the coverage rates move fairly quickly toward $.95$ as $n$ increases. The same conclusions hold for the design with $\mu=0$ and $\beta=2.5$ (f\mbox{}irst block of Table \ref{table:remaining}).

\begin{table}[htbp]
\caption{Mixing proportions (cont'd)}
\begin{center}
\begin{tabular}{rrrrrrrrrrr}

\hline\hline
\multicolumn{3}{l}{} & \multicolumn{2}{c}{BIAS}  & \multicolumn{2}{c}{SD} & \multicolumn{2}{c}{SE/SD} & \multicolumn{2}{c}{CI$95$}  \\ 
\multicolumn{1}{c}{$n$} & \multicolumn{1}{c}{$q_\ell$} & \multicolumn{1}{c}{$q_r$} & \multicolumn{1}{l}{$\lambda_n(0)$} & \multicolumn{1}{l}{$\lambda_n(1)$} & \multicolumn{1}{l}{$\lambda_n(0)$} & \multicolumn{1}{l}{$\lambda_n(1)$} & \multicolumn{1}{l}{$\lambda_n(0)$} & \multicolumn{1}{l}{$\lambda_n(1)$} & \multicolumn{1}{l}{$\lambda_n(0)$} & \multicolumn{1}{l}{$\lambda_n(1)$} \\ 
\hline
\multicolumn{11}{c}{$\mu=0$ and $\beta=2.5$}  \\

$ 1,000$ & $.059$ & $.940$ & $.0066$ & $-.0072$ & $.0722$ & $.0718$ & $1.0151$ & $1.0194$ &$.9646$&$.9652$\\ 
$10,000$ & $.026$ & $.974$ & $.0012$ & $-.0015$ & $.0323$ & $.0326$ & $1.0287$ & $1.0193$ &$.9548$&$.9626$\\ 

\hline

\multicolumn{11}{c}{$\mu=1$ and $\beta=0$}  \\ 
$ 1,000$ &$.059$&$.940$& $.0144$ & $-.0164$ & $.0720$ & $.0728$ & $1.0589$ & $1.0518$ & $.9807$ & $.9810$ \\ 
$10,000$ &$.026$&$.974$& $.0050$ & $-.0048$ & $.0327$ & $.0324$ & $1.0344$ & $1.0449$ & $.9614$ & $.9622$ \\ 

\hline

\multicolumn{11}{c}{$\mu=.5$ and $\beta=0$}  \\  
 
$ 1,000$ & $.090$ & $.910$ & $.0994$ & $-.1017$ & $.0842$ & $.0855$ & $1.1677$ & $1.1599$&$.9416$&$.9406$ \\ 
$10,000$ & $.039$ & $.961$ & $.0671$ & $-.0671$ & $.0358$ & $.0352$ & $1.0815$ & $1.0973$&$.6244$&$.6286$ \\ 

\hline\hline
\end{tabular}
\end{center}
\label{table:remaining}
\end{table}

Now turn to the results for the pure location model with Gaussian errors ($\beta=0$) in Table \ref{table:remaining}, where the tail conditions of Assumption \ref{ass:tailrates} fail. The dif\mbox{}ference between the two designs is the distance between the component distributions (governed by $\mu$). When $\mu=1$, $G$ is centered at $1$ while $H$ is centered at $-1$, so that  $\mu_G-\mu_H=2$. When $\mu=1/2$, $G$ and $H$ are closer to each other: $\mu_G-\mu_H=1$. In the f\mbox{}irst of these designs the bias in the point estimates is somewhat larger than in the skewed designs. Nonetheless, the bias is still small relative to the standard deviation. Furthermore, the coverage of the conf\mbox{}idence intervals displays a similar pattern as before, and is excellent when $n$ is not too small. When we move to the second design the bias increases further. The bias still shrinks to zero as $n$ grows, conf\mbox{}irming that our estimator remains consistent. However,  the bias is not negligible relative  to the standard deviation; the coverage of the conf\mbox{}idence intervals deteriorates as $n$ grows, and inference becomes unreliable.

\begin{figure} 
\caption{Simulation results for $G_n$ (left) and $H_n$ (right) for design $\mu=0$, $\beta=5$ (top) and design $\mu=0$, $\beta=2.5$ (bottom). Each plot contains the mean of the point estimator (solid red line) and the mean of the estimated $95\%$ conf\mbox{}idence bands (dashed blue lines), along with the true curve (solid black line, marked x) and $95\%$ conf\mbox{}idence bands constructed using the Monte Carlo standard deviation (dashed green lines, upper band marked $\bigtriangleup$ and lower band marked $\bigtriangledown$). } 
\includegraphics[width=1\textwidth,height=.75\textheight]{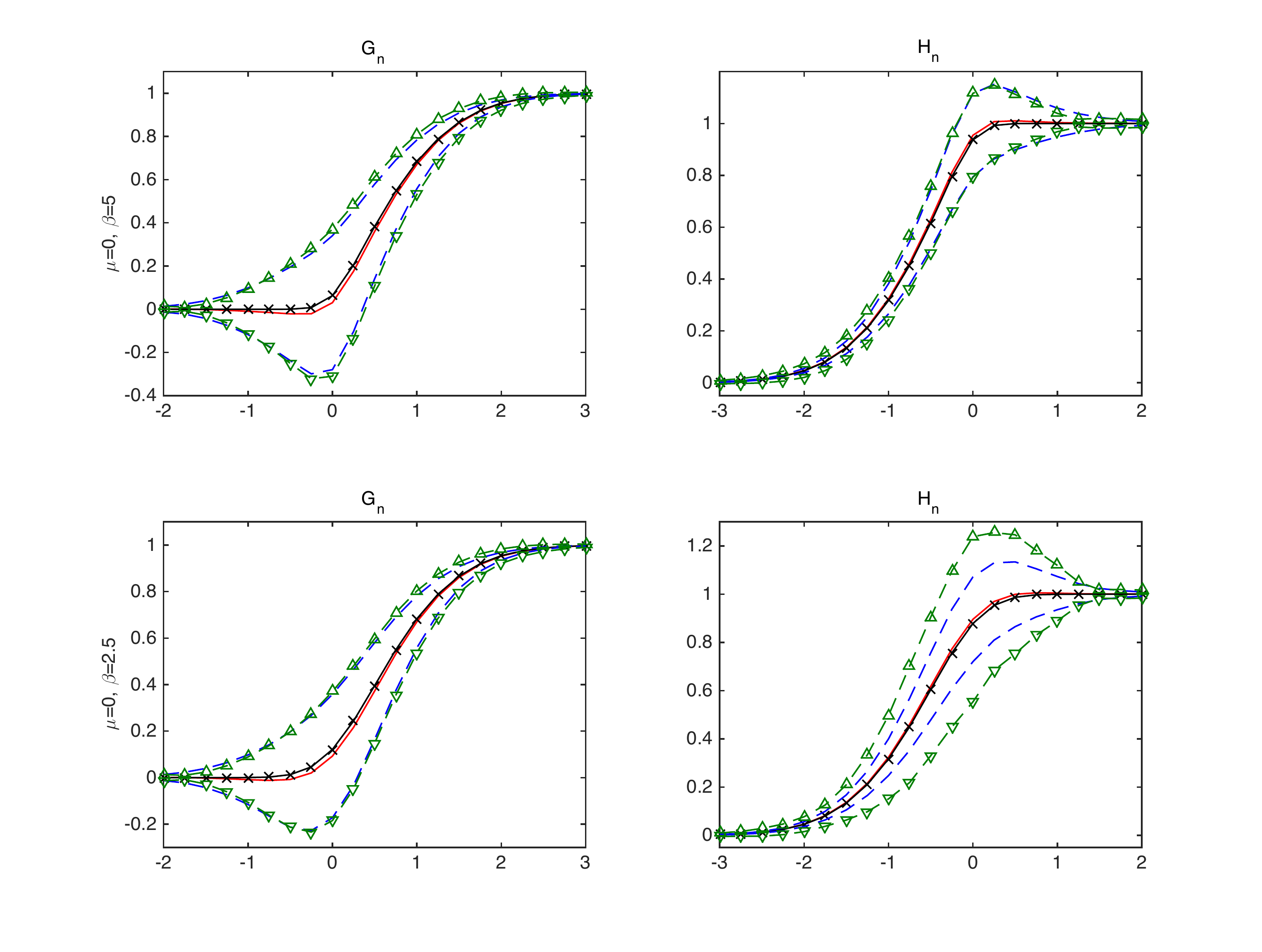}
\label{fig:plot1}
\begin{quote} \footnotesize \vspace{-.40cm}
\end{quote}
\end{figure}

\begin{figure} 
\caption{Simulation results for $G_n$ (left) and $H_n$ (right) for design $\mu=1$, $\beta=0$ (top) and design $\mu=0.5$, $\beta=0$ (bottom). Each plot contains the mean of the point estimator (solid red line) and the mean of the estimated $95\%$ conf\mbox{}idence bands (dashed blue lines), along with the true curve (solid black line, marked x) and $95\%$ conf\mbox{}idence bands constructed using the Monte Carlo standard deviation (dashed green lines, upper band marked $\bigtriangleup$ and lower band marked $\bigtriangledown$). } 
\includegraphics[width=1\textwidth,height=.75\textheight]{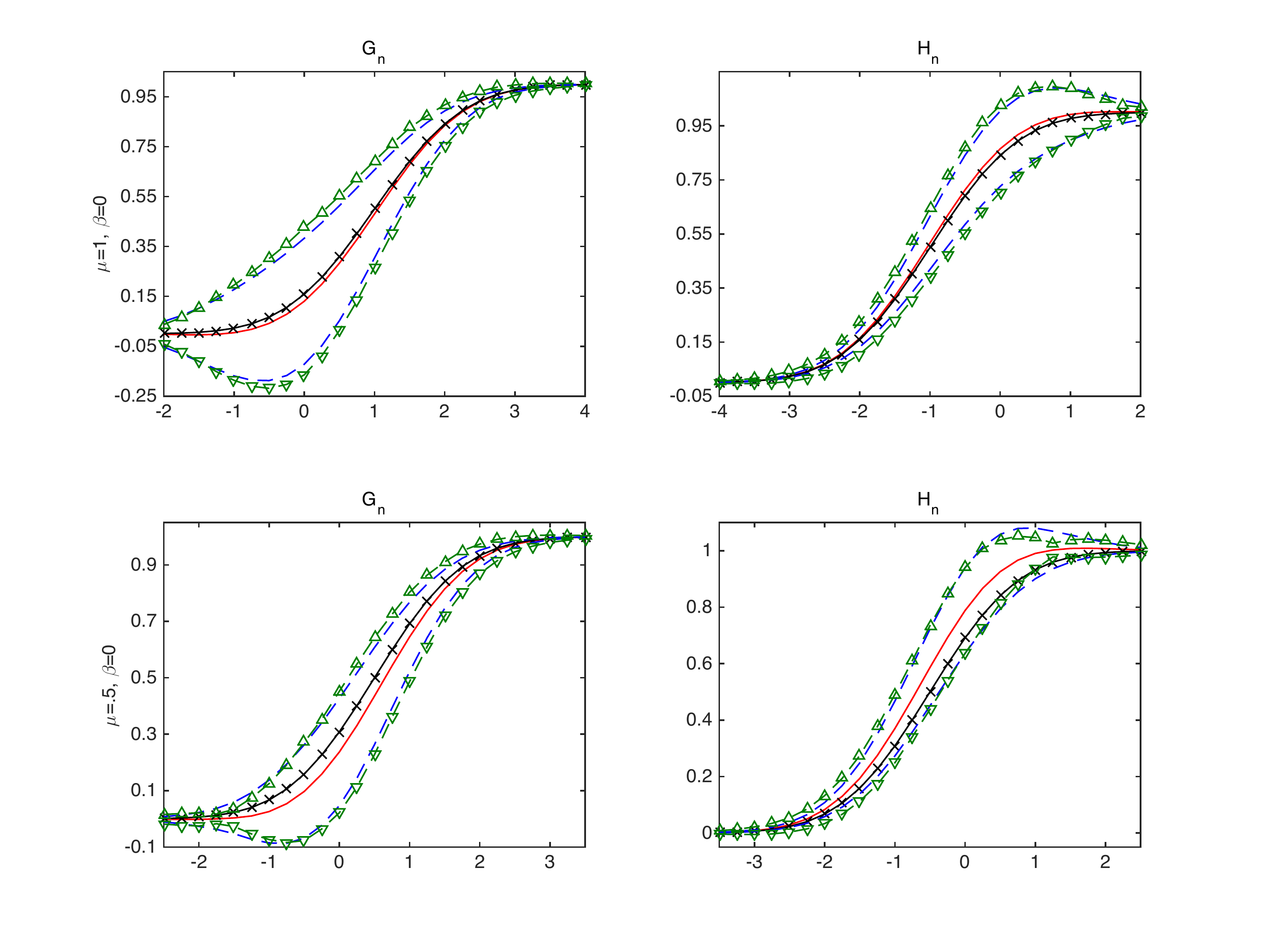}
\label{fig:plot2}
\begin{quote} \footnotesize \vspace{-.40cm}
\end{quote}
\end{figure}

We next turn to the results for the component distributions. For clarity we present the results by means of a series of plots. We provide results for $n=1,000$ for the skewed designs $\mu=0, \beta=5$ and $\mu=0, \beta=2.5$ in Figure \ref{fig:plot1} and for the symmetric designs $\mu=1, \beta=0$ and $\mu=0.5, \beta=0$ in Figure \ref{fig:plot2}. Results for $G_n$ are in the left-side plots. Results for $H_n$ are in the right-side plots. Each plot contains the mean of the point estimates (solid red lines) and the mean of $95\%$ conf\mbox{}idence bounds constructed around it using a plug-in estimator of the asymptotic variance in Theorem \ref{thm:distributions} (dashed blue lines). Each plot also contains the true component distribution (solid black lines, marked x) and the mean of $95\%$ conf\mbox{}idence bounds constructed around the point estimator using the empirical standard deviation over the Monte Carlo replications (dashed green lines, upper band marked $\bigtriangleup$, lower band marked $\bigtriangledown$). We vary the range of the vertical axis across the plots in a given f\mbox{}igure to enhance visibility.

The plots in Figure \ref{fig:plot1} again conf\mbox{}irm our asymptotics. The bias in the point estimators is small across all plots. The asymptotic theory mostly does a good job in capturing the small-sample variability of the point estimators although, when $n$ is small, the standard errors are somewhat too small. In our designs, this underestimation is more severe for $H_n$ than for $G_n$, as is apparent from inspection of the lower-right plot in the figure. Inspection of the full set of results (not reported here) shows that this underestimation vanishes as $n$ grows, again conf\mbox{}irming our asymptotic theory.

The results in Figure \ref{fig:plot2} for the Gaussian location model are in line with our f\mbox{}indings concerning the mixing proportions. In the design where $\mu_G-\mu_H=2$ (upper two plots) our estimators do well in spite of Assumption \ref{ass:tailrates} not holding. When the $\mu_G - \mu_H=1$ (lower two plots), however, the asymptotic bias in $G_n$ and $H_n$ becomes visible. While the variability of the point estimates is correctly captured by our asymptotic-variance estimator, the conf\mbox{}idence bounds settle around an incorrect curve.

\newpage

\section*{Concluding remarks}

We conducted most of our analysis with a mixture of two components. However, some of our results would extend to a version of \eqref{eq:model} with a larger number of components. Suppose that the mixture has $J$ irreducible components, as in
\[
F(y\vert x)=\sum_{j=1}^J  \lambda_j(x) \, G_j(y),
\]
in obvious notation. \cite{HenryKitamuraSalanie2014} showed that the mixture components and mixing proportions are only identif\mbox{}ied up to $J(J-1)$ inequality-constrained real parameters in general. 

Tail dominance restrictions can still be quite powerful. Take $J=3$ for instance, and assume that $G_1$ dominates in the left tail and $G_3$ dominates in the right tail. Then it is easy to adapt the proof of Theorem~\ref{thm:identification} to prove that the behavior of $F(y\vert x)$ in the left tail identif\mbox{}ies the function $\lambda_1$  up to a multiplicative constant, and that the behavior of $F(y\vert x)$ in the right tail identif\mbox{}ies the function $\lambda_3$  up to another multiplicative constant. Imposing the values of the  mixing proportions at one particular value of $x$ would be enough to point identify all elements of the model, for instance; and it would be easy to adapt our estimators and tests
to such a setting. Whether such additonal restrictions are plausible is, of course, highly model-dependent.

\theendnotes

\newpage

\section*{References} \vspace{-.75cm}
\bibliographystyle{chicago}\renewcommand\refname{}
\bibliography{tailrestrictions}

\begin{thebibliography}{}

\bibitem[\protect\citeauthoryear{Acemoglu, Carvalho, Ozdaglar, and
  Tabaz-Salehi}{Acemoglu et~al.}{2012}]{AcemogluEtAl2012}
Acemoglu, D., V.~Carvalho, A.~Ozdaglar, and A.~Tabaz-Salehi (2012).
\newblock The network origins of aggregate fluctuations.
\newblock {\em Econometrica\/}~{\em 80\/}(5), 1977--2016.

\bibitem[\protect\citeauthoryear{Allman, Matias, and Rhodes}{Allman
  et~al.}{2009}]{AllmanMatiasRhodes2009}
Allman, E.~S., C.~Matias, and J.~A. Rhodes (2009).
\newblock Identifiability of parameters in latent structure models with many
  observed variables.
\newblock {\em Annals of Statistics\/}~{\em 37}, 3099--3132.

\bibitem[\protect\citeauthoryear{Andrews and Schafgans}{Andrews and
  Schafgans}{1998}]{AndrewsSchafgans1998}
Andrews, D. W.~K. and M.~M.~A. Schafgans (1998).
\newblock Semiparametric estimation of the intercept of a sample selection
  model.
\newblock {\em Review of Economic Studies\/}~{\em 65}, 497--517.

\bibitem[\protect\citeauthoryear{Arkolakis, Costinot, and
  Rodriguez-Clare}{Arkolakis
  et~al.}{2012}]{ArkolakisCostinotRodriguezClare2012}
Arkolakis, C., A.~Costinot, and A.~Rodriguez-Clare (2012).
\newblock New trade models, same old gains?
\newblock {\em American Economic Review\/}~{\em 102}, 94--130.

\bibitem[\protect\citeauthoryear{Atkinson, Piketty, and Saez}{Atkinson
  et~al.}{2011}]{AtkinsonPikettySaez2011}
Atkinson, A.~B., T.~Piketty, and E.~Saez (2011).
\newblock Top incomes in the long run of history.
\newblock {\em Journal of Economic Literature\/}~{\em 49}, 3--71.

\bibitem[\protect\citeauthoryear{Azzalini}{Azzalini}{1985}]{Azzalini1985}
Azzalini, A. (1985).
\newblock A class of distributions which includes the normal ones.
\newblock {\em Scandinavian Journal of Statistics\/}~{\em 12}, 171--178.

\bibitem[\protect\citeauthoryear{Bollinger}{Bollinger}{1996}]{Bollinger1996}
Bollinger, C.~R. (1996).
\newblock Bounding mean regressions when a binary regressor is mismeasured.
\newblock {\em Journal of Econometrics\/}~{\em 73}, 387--399.

\bibitem[\protect\citeauthoryear{Bonhomme, Jochmans, and Robin}{Bonhomme
  et~al.}{2014}]{BonhommeJochmansRobin2014b}
Bonhomme, S., K.~Jochmans, and J.-M. Robin (2014).
\newblock Estimating multivariate latent-structure models.
\newblock {\em Annals of Statistics}, forthcoming.

\bibitem[\protect\citeauthoryear{Bonhomme, Jochmans, and Robin}{Bonhomme
  et~al.}{2016}]{BonhommeJochmansRobin2014a}
Bonhomme, S., K.~Jochmans, and J.-M. Robin (2016).
\newblock Nonparametric estimation of f\mbox{}inite mixtures from repeated
  measurements.
\newblock {\em Journal of the Royal Statistical Society, Series B\/}~{\em 78},
  211--229.

\bibitem[\protect\citeauthoryear{Bordes, Mottelet, and Vandekerkhove}{Bordes
  et~al.}{2006}]{BordesMotteletVandekerkhove2006}
Bordes, L., S.~Mottelet, and P.~Vandekerkhove (2006).
\newblock Semiparametric estimation of a two-component mixture model.
\newblock {\em Annals of Statistics\/}~{\em 34}, 1204--1232.

\bibitem[\protect\citeauthoryear{Capitanio}{Capitanio}{2010}]{Capitanio2010}
Capitanio, A. (2010).
\newblock On the approximation of the tail probability of the scalar
  skew-normal distribution.
\newblock {\em METRON\/}~{\em 68}, 299--308.

\bibitem[\protect\citeauthoryear{Carroll, Ruppert, Stefanski, and
  Crainiceanu}{Carroll et~al.}{2006}]{CarrollRuppertStefanskiCrainiceanu2006}
Carroll, R.~J., D.~Ruppert, L.~A. Stefanski, and C.~Crainiceanu (2006).
\newblock {\em Measurement Error in Nonlinear Models: A Modern Perspective}.
\newblock Chapman and Hall, CRC Press.

\bibitem[\protect\citeauthoryear{D'Haultf{\oe}uille and
  F\'evrier}{D'Haultf{\oe}uille and
  F\'evrier}{2015}]{DhaultfoeuilleFevrier2014}
D'Haultf{\oe}uille, X. and P.~F\'evrier (2015).
\newblock Identif\mbox{}ication of mixture models using support variations.
\newblock {\em Journal of Econometrics\/}~{\em 189}, 70--82.

\bibitem[\protect\citeauthoryear{D'Haultfoeuille and Maurel}{D'Haultfoeuille
  and Maurel}{2013}]{DhaultfoeuilleMaurel2013}
D'Haultfoeuille, X. and A.~Maurel (2013).
\newblock Another look at identification at inf\mbox{}inity of sample selection
  models.
\newblock {\em Econometric Theory\/}~{\em 29}, 213--224.

\bibitem[\protect\citeauthoryear{Einmahl}{Einmahl}{1992}]{Einmahl1992}
Einmahl, J. (1992).
\newblock Limit theorems for tail processes with application to intermediate
  quantile estimation.
\newblock {\em Journal of Statistical Planning and Inference\/}~{\em 32},
  137--145.

\bibitem[\protect\citeauthoryear{Einmahl and Mason}{Einmahl and
  Mason}{1997}]{EinmahlMason1997}
Einmahl, U. and D.~Mason (1997).
\newblock Gaussian approximation of local empirical processes indexed by
  functions.
\newblock {\em Probability Theory and Related Fields\/}~{\em 107}, 283--311.

\bibitem[\protect\citeauthoryear{Frisch}{Frisch}{1934}]{Frisch1934}
Frisch, R. (1934).
\newblock Statistical confluence analysis by means of complete regression
  systems.
\newblock Technical Report~5, University of Oslo, Economics Institute, Oslo,
  Norway.

\bibitem[\protect\citeauthoryear{Gabaix}{Gabaix}{2009}]{Gabaix2009}
Gabaix, X. (2009).
\newblock Power laws in economics and finance.
\newblock {\em Annual Review of Economics\/}~{\em 1}, 255--294.

\bibitem[\protect\citeauthoryear{Gassiat and Rousseau}{Gassiat and
  Rousseau}{2016}]{GassiatRousseau2014}
Gassiat, E. and J.~Rousseau (2016).
\newblock Nonparametric f\mbox{}inite translation hidden {M}arkov models and
  extensions.
\newblock {\em Bernoulli\/}~{\em 22}, 193--212.

\bibitem[\protect\citeauthoryear{Ghysels, Harvey, and Renault}{Ghysels
  et~al.}{1996}]{GhyselsHarveyRenault1996}
Ghysels, E., A.~Harvey, and E.~Renault (1996).
\newblock Stochastic volatility.
\newblock In G.~S. Maddala and C.~R. Rao (Eds.), {\em Handbook of Statistics
  Volume 14: Statistical Methods in Finance}. Elsevier.

\bibitem[\protect\citeauthoryear{Hall and Zhou}{Hall and
  Zhou}{2003}]{HallZhou2003}
Hall, P. and X.-H. Zhou (2003).
\newblock Nonparametric identif\mbox{}ication of component distributions in a
  multivariate mixture.
\newblock {\em Annals of Statistics\/}~{\em 31}, 201--224.

\bibitem[\protect\citeauthoryear{Hamilton}{Hamilton}{1989}]{Hamilton1989}
Hamilton, J.~D. (1989).
\newblock A new approach to the analysis of nonstationary times series and the
  business cycle.
\newblock {\em Econometrica\/}~{\em 57}, 357--384.

\bibitem[\protect\citeauthoryear{Heckman}{Heckman}{1974}]{Heckman1974}
Heckman, J.~J. (1974).
\newblock Shadow prices, market wages, and labor supply.
\newblock {\em Econometrica\/}~{\em 42}, 679--694.

\bibitem[\protect\citeauthoryear{Heckman}{Heckman}{1990}]{Heckman1990}
Heckman, J.~J. (1990).
\newblock Varieties of selection bias.
\newblock {\em American Economic Review\/}~{\em 80}, 313--318.

\bibitem[\protect\citeauthoryear{Henry, Kitamura, and Salani\'e}{Henry
  et~al.}{2010}]{HenryKitamuraSalanie2010}
Henry, M., Y.~Kitamura, and B.~Salani\'e (2010).
\newblock Identifying f\mbox{}inite mixtures in econometric models.
\newblock Cowles Foundation Discussion Paper 1767.

\bibitem[\protect\citeauthoryear{Henry, Kitamura, and Salani\'e}{Henry
  et~al.}{2014}]{HenryKitamuraSalanie2014}
Henry, M., Y.~Kitamura, and B.~Salani\'e (2014).
\newblock Partial identif\mbox{}ication of f\mbox{}inite mixtures in
  econometric models.
\newblock {\em Quantitative Economics\/}~{\em 5}, 123--144.

\bibitem[\protect\citeauthoryear{Hu and Schennach}{Hu and
  Schennach}{2008}]{HuSchennach2008}
Hu, Y. and S.~M. Schennach (2008).
\newblock Instrumental variable treatment of nonclassical measurement error
  models.
\newblock {\em Econometrica\/}~{\em 76}, 195--216.

\bibitem[\protect\citeauthoryear{Jochmans, Henry, and Salani\'e}{Jochmans
  et~al.}{2014}]{JochmansHenrySalanie2014}
Jochmans, K., M.~Henry, and B.~Salani\'e (2014).
\newblock Inference on mixtures under tail restrictions.
\newblock Discussion Paper No 2014-01, Department of Economics, Sciences Po.

\bibitem[\protect\citeauthoryear{Kasahara and Shimotsu}{Kasahara and
  Shimotsu}{2009}]{KasaharaShimotsu2009}
Kasahara, H. and K.~Shimotsu (2009).
\newblock Nonparametric identif\mbox{}ication of f\mbox{}inite mixture models
  of dynamic discrete choices.
\newblock {\em Econometrica\/}~{\em 77}, 135--175.

\bibitem[\protect\citeauthoryear{Khan and Tamer}{Khan and
  Tamer}{2010}]{KhanTamer2010}
Khan, S. and E.~Tamer (2010).
\newblock Irregular identif\mbox{}ication, support conditions and inverse
  weight estimation.
\newblock {\em Econometrica\/}~{\em 78}, 2021--2042.

\bibitem[\protect\citeauthoryear{Lewbel}{Lewbel}{2007}]{Lewbel2006}
Lewbel, A. (2007).
\newblock Estimation of average treatment effects with misclassification.
\newblock {\em Econometrica\/}~{\em 75}, 537--551.

\bibitem[\protect\citeauthoryear{Mahajan}{Mahajan}{2006}]{Mahajan2006}
Mahajan, A. (2006).
\newblock Identif\mbox{}ication and estimation of regression models with
  misclassification.
\newblock {\em Econometrica\/}~{\em 74}, 631--665.

\bibitem[\protect\citeauthoryear{Schwarz and Van~Bellegem}{Schwarz and
  Van~Bellegem}{2010}]{SchwarzVanBellegem2010}
Schwarz, M. and S.~Van~Bellegem (2010).
\newblock Consistent density deconvolution under partially known error
  distribution.
\newblock {\em Statistics and Probability Letters\/}~{\em 80}, 236--241.

\bibitem[\protect\citeauthoryear{Shimer and Smith}{Shimer and
  Smith}{2000}]{ShimerSmith2000}
Shimer, R. and L.~Smith (2000).
\newblock Assortative matching and search.
\newblock {\em Econometrica\/}~{\em 68}, 343--369.

\end{thebibliography}

\end{document}